\documentclass[12pt]{article}

\usepackage{amsmath, amssymb, amsthm, mathrsfs}
\usepackage[margin=1.25in]{geometry}
\usepackage[utf8]{inputenc}
\usepackage{enumitem}
\usepackage{booktabs}
\usepackage{graphicx}
\usepackage{natbib}
\usepackage{tikz}
\usetikzlibrary{patterns, positioning, calc, decorations.pathreplacing}
\usepackage[colorlinks=true, citecolor=black, linkcolor=black, urlcolor=black]{hyperref}
\usepackage{subcaption}
\usepackage{setspace}

\theoremstyle{definition}
\newtheorem{definition}{Definition}
\newtheorem{assumption}{Assumption}

\newtheorem{remark}{Remark}

\theoremstyle{plain}
\newtheorem{theorem}{Theorem}

\newtheorem{proposition}[theorem]{Proposition}
\newtheorem{corollary}[theorem]{Corollary}

\newcommand{\Acal}{\mathcal{A}}        

\newcommand{\Prb}{\mathbb{P}}
\newcommand{\Ind}{\mathbf{1}}          

\DeclareMathOperator{\rank}{rank}
\DeclareMathOperator{\diag}{diag}

\newcommand{\dto}{\xrightarrow{d}}      

\newcommand{\dOneFrac}{100.0}
\newcommand{\dOneCount}{9,831}
\newcommand{\dOneTotal}{9,831}

\newcommand{\gidDeffPlusOne}{12}
\newcommand{\gidNclusters}{50}
\newcommand{\gidCellsRankOk}{41}

\title{A Random Rule Model}
\author{Avner Seror\thanks{Aix-Marseille University, CNRS, AMSE, Marseille, France. Email: avner.seror@univ-amu.fr. I am grateful to Marina Agranov, Benjamin Enke, En Hua Hu, Dotan Persitz, John Quah, and Thierry Verdier for helpful discussions. I acknowledge funding from the French government under the ANR JCJC ``BIAS" project (reference: ANR-25-CE26-7164-01) and the ``France 2030'' investment plan managed by the French National Research Agency (reference: ANR-17-EURE-0020) and from the Excellence Initiative of Aix-Marseille University -- A*MIDEX. Mistakes are my own. Refine.ink was used to check the paper for consistency and clarity. The author used Claude and ChatGPT for code development and editorial revisions; all scientific content, modeling choices, and results were produced and verified by the author.}}
\date{April 14, 2026}

\begin{document}
\onehalfspacing
\maketitle

\begin{abstract}
We model stochastic choice as environment-dependent switching among a small library of deterministic decision rules. A Random Rule Model generates menu-level choice probabilities via named, interpretable rules weighted by observable menu characteristics. Identification has a two-step structure: within-feature decisive-side variation identifies relative rule weights; cross-feature richness identifies the gate. Applied to binary lottery choices, the estimated weights concentrate on a small subset of rules and shift systematically with complexity and dispersion asymmetry. The model closes nearly all of the prediction gap to a flexible neural-network benchmark, while remaining interpretable, restrictive under permutation diagnostics, and portable to an independent dataset.
\end{abstract}

\vspace{0.5em}

\vspace{0.3em}

\section{Introduction}\label{sec:intro}

Stochastic choice data are often interpreted through latent utility heterogeneity, random utility shocks, or menu-dependent consideration. An alternative possibility is that observed choice frequencies reflect systematic switching across decision procedures. In many environments, agents may sometimes screen on extreme outcomes, sometimes compare modal consequences, sometimes focus on salient payoff differences, and sometimes follow simpler attention-based procedures. The central object in such settings is not only the average choice probability attached to a menu, but the procedural heterogeneity that determines which decision procedures are more prevalent in which environments.

This paper studies stochastic choice through that lens. We develop a Random Rule Model (RRM) in which behavior is generated by random switching among a small library of transparent, deterministic decision rules. Each rule maps a menu directly into a binary recommendation - ``choose Left'' or ``choose Right'' - or declares itself inapplicable when it cannot rank the two options; stochastic choice arises because the relative weight placed on these rules varies with observable features of the menu. The resulting model yields an explicit procedural decomposition of behavior: changes in choice frequencies across environments are interpreted as changes in the prevalence of named lines of conduct rather than as shifts in an unobserved reduced-form predictor.

This representation is attractive for several reasons. First, it provides a tractable language for procedural heterogeneity in large choice environments. Rather than forcing a single valuation model to govern all menus, the model allows different procedures to matter in different regions of the environment space. Second, it remains economically interpretable: the primitive objects are transparent rules with clear behavioral content. Third, it imposes discipline. The model can fit choice frequencies only through the patterns of agreement and disagreement generated by the rule library and through systematic variation in the corresponding rule weights. Fourth, it does not require transitivity or other consistency conditions on the aggregate choice function: because each rule is applied independently and the mixture weights vary freely with menu features, the model can accommodate systematic violations of classical rationality axioms that are well-documented in stochastic choice data.

The rule weights depend on observable features of the menu, and not all rules deliver a strict recommendation on every menu - the model restricts behavioral responsibility to rules that are actually decisive, ensuring that the decomposition attributes weight only where a rule provides a clear-cut ranking. The formal mechanism is a softmax gate over rule-specific indices.

Both the RRM and Random Utility Models (RUM) are instances of a broader class: stochastic selection over decision processes. In a RUM, the selected processes are utility maximizers - rich, parametric objects whose heterogeneity generates stochastic choice. In the RRM, the selected processes are coarse, parameter-free heuristics that need not be optimization-based. This is a substantive modeling distinction about the nature of behavioral heterogeneity: preference heterogeneity versus procedural heterogeneity \citep{simon1955,rubinstein88}.

The distinction matters at three levels. First, the objects recovered are different. A RUM identifies a mixing distribution over utility functions. The RRM identifies which named procedures are used, how much, and how their prevalence shifts with the environment - the kind of procedural decomposition that has no natural RUM counterpart. Second, the identification problem is different. In a RUM, the decision processes are latent and must be inferred from choice frequencies alone. In the RRM, the rules are deterministic and their recommendations are directly observable; identification exploits the fact that different rules switch sides across menus, a novel identifying force. Third, because the primitives are simple and parameter-free, the model is tightly disciplined: it cannot absorb structureless variation, as we verify with a permutation-fit diagnostic.

The main theoretical question is whether the latent procedural weights are genuinely identified from menu-level choice frequencies. We show that they are, under a two-layer structure. In the first layer, consider menus that share the same value of the feature vector used by the gate (a possibly coarsened summary of the full menu description). These menus also share the same latent rule weights, but they may differ in which rules recommend which option - some rules switch sides across menus even when the gate features are held fixed. This rule-switching variation pins down the relative importance of the rules at a given set of characteristics: if enough rules switch sides across enough menus, the data uniquely determine how much weight each rule carries. In the second layer, once relative rule weights have been recovered at sufficiently many distinct menu characteristics, the law governing how those weights change with menu features is itself identified. Global identification thus separates cleanly into a within-characteristic step that uses rule switching to recover relative weights, and an across-characteristic step that maps those weights back into the gate.

The identification theorem is constructive: it implies a two-step estimator in which normalized rule weights are first recovered cell by cell from the linear restrictions generated by rule switching, and the gate parameters are then estimated from the cross-cell variation in recovered log-weights. We establish consistency and asymptotic normality for this estimator, yielding inference for economically meaningful functionals of the decomposition.

For the main empirical analysis, however, we estimate the model by minimizing out-of-sample mean squared error (MSE). This avoids the binning of the feature space into discrete cells that the two-step procedure requires and places the rule-gating model on the same predictive footing as the neural-network benchmarks against which it is evaluated. The cellwise identification theory guarantees that the rule weights recovered at each menu characteristic have structural content - they are pinned down by the data, not merely by the predictive objective. When the gate specification is correct, both estimators target the same population decomposition; the two-step estimator provides a built-in diagnostic for this restriction.

We apply the framework to the \texttt{choices13k} dataset of \citet{peterson2021}, which contains thousands of distinct binary lottery menus with repeated observations per menu. We construct a library of twelve parameter-free decision rules spanning canonical behavioral mechanisms - extreme-outcome screening, modal-outcome comparisons, salience-based comparisons, regret- and disappointment-type comparisons, and limited-consideration-style defaults - and estimate the model using the same cross-validation protocol as \citet{peterson2021}. We first verify the identification conditions on the data, then report the MSE-fitted behavioral decomposition: responsibility weights (the average share of predicted choice attributed to each rule), ablation indices (the predictive deterioration when a rule is removed), and concentration measures (how dispersed rule usage is across the library). The two-step econometric estimates follow as a cross-check. We then study how rule responsibility shifts with interpretable menu characteristics such as tradeoff complexity and dispersion asymmetry, and conclude with predictive validation and portability to an independent dataset.

The main empirical findings are as follows. Both conditions of the global identification theorem are satisfied on the Choices13k data. The estimated weights concentrate on a small subset of rules: salience-based comparisons, regret-type comparisons, and modal-outcome screening carry the largest non-attention responsibility shares, while the attention channel absorbs roughly half of effective responsibility, consistent with substantial inattention in online data collection. Rule responsibility shifts systematically with menu primitives: along tradeoff complexity, regret declines steeply and the mixture reallocates toward attention and salience; along dispersion asymmetry, regret declines while attention defaults rise.

On predictive performance, rule-gating (MSE $0.0117$) closes about $97\%$ of the gap between Neural EU ($0.0222$) and the strongest benchmark, the Mixture of Theories (MOT, $0.0114$). Rule-gating is essentially fully restrictive under permutation-fit diagnostics, while MOT absorbs substantial structureless variation. This distinction is confirmed by out-of-sample portability to an independent dataset \citep{erev_ert_plonsky_2017,plonsky2017_cpc18}: rule-gating transfers well while MOT deteriorates sharply. A battery of robustness checks - cross-fitted library selection, attention-rule removal, activity-discipline variation, placebo tests, and feature-based baselines - confirms the stability of the decomposition.

\paragraph{Related literature.}
Our approach relates to the literature on stochastic choice. In random utility models (RUM), choice probabilities are generated by heterogeneity in latent utilities \citep{mcfadden2006,McFaddenRichter1991,stoye2018_ecma}. A related literature accounts for limited, random, or menu-dependent consideration, jointly recovering preferences and attention/consideration mechanisms from stochastic choice data \citep{manzini_mariotti2014,cattaneo_ram,aguiar2023}. As discussed above, both RUM and RRM belong to the broader class of stochastic selection over decision processes, but differ in the nature of the selected primitives (optimization-based vs.\ heuristic), the recovered objects (utility distributions vs.\ procedural responsibilities), and the identification structure (latent preference heterogeneity vs.\ observable rule switching). The relationship to the attention/consideration literature is analogous: the RRM does not model consideration sets but rather procedural switching conditional on rule applicability.

Our identification argument connects to the econometrics of latent heterogeneity. The first-stage identification - recovering mixture weights from linear restrictions generated by rule switching - is structurally related to identification in finite mixture models with instrument-like variation \citep{henry_kitamura_salanie2014}. The novelty is that the identifying variation comes from the discrete patterns of decisive-side recommendations, which are observable and deterministic, rather than from excluded instruments or panel structure.

The paper also connects to the broader tradition of modeling choice as the outcome of procedures under bounded rationality and simplicity constraints \citep{simon1955,rubinstein88,aumann2008,manzini2007,manzini2012,cherepano2013,salant_rubinstein2008}, and to the behavioral literature on menu-dependent perception and evaluation \citep{kahneman_tversky1979,tversky_kahneman1992,bordalo2012_risk,loomessugden1986,masatlioglu2012}. Our mechanism results relate to the literature on complexity and rule selection \citep{arrieta2024,halevy2024,puri2025,enke2023}; closest to our analysis, \citet{enke2023} construct menu-based complexity indices and show that ``tradeoff complexity'' predicts behavioral signatures in large-scale choice data. We build on this by studying how complexity and dispersion asymmetry are associated with systematic reallocation of procedure weights inside an interpretable random-rule mixture. Our approach builds on the decision-rule primitives of \citet{seror2026}; the present paper adds the gating mechanism, the identification theory, and the econometric framework.

Finally, our approach also speaks to the recent agenda of empirical machine learning in economics \citep{liang2025}. Our completeness normalization follows recent work benchmarking economic structure against flexible ML baselines
(e.g., \citealp{peysakhovich2017,fudenberg2020,fudenberg2022,fudenberg2022_jpe}); our restrictiveness diagnostic
applies the logic of \citet{fudenberg26}; and we implement a portability exercise targeting external validity.
Finally, a growing literature documents strong out-of-sample performance in risky choice (e.g., \citealp{erev2011,erev_ert_plonsky_2017,plonsky2016,peterson2021}).
Our contribution is to show that a disciplined mixture of interpretable procedures can recover a large share of these predictive gains while producing mechanism-level objects that are directly interpretable in economic terms.

\section{Model}\label{sec:model}

We develop the model and the identification theory for binary lottery menus, which is the setting of our empirical application. Appendix~\ref{app:general_framework} presents the fully general framework covering arbitrary menu sizes, general dominance preorders, and multinomial choice. The two treatments share the same structure - a grouped-softmax representation, activity-conditional normalization, and verifiable conditions for global identification - with the main difference being notational: log-odds become log-ratios and two-sided menus become multi-sided menus. The binary exposition thus conveys the essential ideas with lighter notation.

\subsection{Setup}\label{subsec:setup}

A \emph{menu} (choice set) is a binary set $A=\{L^1,L^2\}$ of lotteries, each lottery being a finite-support distribution over monetary outcomes. We label $L^1$ the left choice and $L^2$ the right choice. Let $\Acal_{\mathrm{full}}$ denote the full set of menus in the dataset. Each $A\in\Acal_{\mathrm{full}}$ is observed
$n(A)\ge 1$ times. Let $c_m(A)\in\{L^1,L^2\}$ be the realized choice in repetition $m\in\{1,\dots,n(A)\}$. The menu-level left-choice frequency is
\[
\hat p(A)\ :=\ \frac{1}{n(A)}\sum_{m=1}^{n(A)} \Ind\{c_m(A)=L^1\}\in[0,1].
\]
We interpret $\hat p(A)$ as estimating a population choice probability $p(A):=\Prb(L^1\mid A)$.  A predictor is a function $g:\Acal_{\mathrm{full}}\to[0,1]$ intended to approximate $p(\cdot)$.

The dataset is a collection of observations $D=\{(c_t,A_t)\}_{t\in\mathcal T}$, where $\mathcal T:=\{1,\dots,T\}$.

\subsection{Rule library}\label{subsec:library}

We model behavior as the outcome of a small library of deterministic decision rules.
Let $\mathcal F$ be a finite collection of rules. Each rule $f\in\mathcal F$ maps a menu to a recommendation or abstention. Formally, for each menu $A=\{L^1,L^2\}$ and rule~$f$, define the \emph{activity indicator}
\begin{equation}\label{eq:Af_def}
A_f(A)\ :=\ \begin{cases} 1 & \text{if }f\text{ delivers a strict recommendation at }A,\\
0 & \text{otherwise,}\end{cases}
\end{equation}
and, when $A_f(A)=1$, the \emph{left-recommendation indicator}
\begin{equation}\label{eq:Lf_def}
L_f(A)\ :=\ \begin{cases} 1 & \text{if }f\text{ recommends }L^1,\\
0 & \text{if }f\text{ recommends }L^2.\end{cases}
\end{equation}
Thus each rule is a deterministic mapping from menus to $\{\text{Left, Right, abstain}\}$. Rules are parameter-free: their recommendations depend only on the menu, not on any estimated quantity.

\begin{definition}[Activity]\label{def:activity}
A rule $f\in\mathcal F$ is \emph{active} at menu $A$ if $A_f(A)=1$, and \emph{inactive} otherwise.
\end{definition}

The pair $(A_f(A),\,L_f(A))$ is the primitive the model requires from each rule. The CA-RRM and identification theory below are stated entirely in terms of these indicators - they do not depend on how rules are implemented internally. In the application (Section~\ref{subsec:library_details}), we construct a concrete library of twelve rules by specifying, for each rule, a \emph{perceived-lottery transformation} that maps the menu into a simplified comparison; the activity and recommendation indicators are then derived by checking First-Order Stochastic Dominance (FSD) of the perceived lotteries. This two-step construction - transform, then compare via FSD - is an implementation device that provides a uniform way to define $(A_f,L_f)$ for diverse heuristics; it is not part of the abstract model.

\subsection{Conditional on Activity Random Rule Model (CA-RRM)}\label{subsec:carrm}

A Conditional on Activity Random Rule Model (CA-RRM) assigns to each menu $A$ a vector
$q(A)=\{q_f(A)\}_{f\in \mathcal F}\in\Delta^{F-1}$, where $q_f(A)$ is the propensity to apply rule $f$ at $A$. Since not all rules are decisive at all menus,
the model conditions on \emph{activity}. We define the left mass and active mass:
\[
\ell(A;q)=\sum_{f\in\mathcal F} q_f(A)\,L_f(A),
\qquad
m(A;q)=\sum_{f\in\mathcal F} q_f(A)\,A_f(A).
\]
The CA-RRM predicted left-choice probability is
\begin{equation}\label{eq:carrm_prob}
g_{\mathrm{CA}}(A)
=\frac{\ell(A;q)}{\max\{m(A;q),m_{\min}\}}
\in[0,1],
\end{equation}
where $m_{\min}>0$ is a small guard constant to prevent numerical instability if the total active mass is near zero. In practice however, we have two  ``attention'' rules in the library $\mathcal F$ that can always be active, so $m(A;q)$ is typically strictly positive.

\subsection{Rule-gating parameterization}\label{subsec:gating}

CA-RRM becomes a \emph{rule-gating predictor} when $q(A)$ is parameterized by a differentiable gate.
Let $\psi(A)\in\mathbb R^d$ be a fixed-length encoding of the menu (defined precisely in Section~\ref{subsec:encoding}).
We use a softmax gate:
\begin{equation}\label{eq:softmax_gate}
q_f(A;\theta)=\frac{\exp(\alpha_f+\beta_f^\top \psi(A))}
{\sum_{g\in\mathcal F} \exp(\alpha_g+\beta_g^\top \psi(A))},
\end{equation}
where $\theta=\{(\alpha_f,\beta_f)\}_{f\in\mathcal F}$. The resulting predictor is
$g_{\mathrm{RG}}(A;\theta)=g_{\mathrm{CA}}(A;q(\cdot;\theta))$ as in \eqref{eq:carrm_prob}.


Combining \eqref{eq:carrm_prob} with \eqref{eq:softmax_gate}, we have
\[
g_{\mathrm{RG}}(A;\theta)=\frac{\sum_{f\in\mathcal F} q_f(A;\theta)\,L_f(A)}
{\sum_{f\in\mathcal F} q_f(A;\theta)\,A_f(A)}\qquad (\text{up to the small guard }m_{\min}).
\]
We define the \emph{conditional-on-activity weights}
\begin{equation}\label{eq:qtilde_def}
\tilde q_f(A;\theta)
:=\frac{q_f(A;\theta)\,A_f(A)}{\sum_{g\in\mathcal F} q_g(A;\theta)\,A_g(A)}\in[0,1],
\qquad \sum_{f\in\mathcal F}\tilde q_f(A;\theta)=1
\end{equation}
whenever the denominator is positive.\footnote{In our implementation, rules \textsc{A1} and \textsc{A2} are always included and are decisive for essentially all menus, so the denominator is typically well away from zero.}
Then the CA-RRM prediction admits the simple mixture form
\begin{equation}\label{eq:mixture_form}
g_{\mathrm{RG}}(A;\theta)=\sum_{f\in\mathcal F}\tilde q_f(A;\theta)\,L_f(A).
\end{equation}
From \eqref{eq:mixture_form}, a CA-RRM operates as a mixture over deterministic rule recommendations, after filtering out rules that are inactive at $A$.

\section{Identification via Rule Switching}\label{sec:identification_switching}

This section clarifies when the gate parameters
$\theta=\{(\alpha_f,\beta_f)\}_{f\in\mathcal F}$
are pinned down by menu-level choice frequencies.
The identification argument has two layers. The first layer identifies, at a fixed feature value,
the vector of relative rule weights. The second layer uses variation across feature values to recover
the affine gate parameters. The key identifying force throughout is \emph{rule switching}: because
different rules recommend different sides in different menus, the decisive-side pattern varies across
the menu support, and this variation pins down the latent rule weights.

\paragraph{From CA-RRM to a collapsed binary logit index.}
For a menu $A=\{L^1,L^2\}$, define the decisive-left and decisive-right sets
\[
\mathcal I_L(A)\ :=\ \{f\in\mathcal F:\ A_f(A)=1,\ L_f(A)=1\},
\qquad
\mathcal I_R(A)\ :=\ \{f\in\mathcal F:\ A_f(A)=1,\ L_f(A)=0\}.
\]
Ignoring the numerical guard $m_{\min}$ (inactive in our implementation), the softmax gate
\eqref{eq:softmax_gate} combined with the CA normalization yields
\begin{equation}\label{eq:carrrm_expform_rank}
g_{\mathrm{RG}}(A;\theta)
=
\frac{\sum_{f\in \mathcal I_L(A)} \exp(\alpha_f+\beta_f^\top\psi(A))}
{\sum_{f\in \mathcal I_L(A)\cup\mathcal I_R(A)} \exp(\alpha_f+\beta_f^\top\psi(A))}.
\end{equation}
Define the inclusive values
\[
V_L(A;\theta):=\log\!\sum_{f\in \mathcal I_L(A)} \exp(\alpha_f+\beta_f^\top\psi(A)),\qquad
V_R(A;\theta):=\log\!\sum_{f\in \mathcal I_R(A)} \exp(\alpha_f+\beta_f^\top\psi(A)).
\]
Whenever $\mathcal I_L(A)$ and $\mathcal I_R(A)$ are both nonempty,
\begin{equation}\label{eq:logit_index_identity_rank}
\log\frac{g_{\mathrm{RG}}(A;\theta)}{1-g_{\mathrm{RG}}(A;\theta)}
\;=\;
V_L(A;\theta)-V_R(A;\theta).
\end{equation}

\paragraph{Observed object.}
Let
\[
\mathcal A_{2}:=\{A\in\mathcal A_{\mathrm{full}}:\ \mathcal I_L(A)\neq\emptyset,\ \mathcal I_R(A)\neq\emptyset\}
\]
denote the set of \emph{two-sided} menus.
For $A\in\mathcal A_2$, define the population log-odds
\begin{equation}\label{eq:delta_rank}
\delta(A)\ :=\ \log\frac{p(A)}{1-p(A)}\ =\ V_L(A;\theta)-V_R(A;\theta),
\end{equation}
where $p(A)=\Pr(L^1\mid A)$.

\paragraph{Normalization and parameter vector.}
Because the softmax gate is invariant to adding a common constant to all $\alpha_f$ and a common vector to all $\beta_f$,
we impose a location normalization: fix one baseline rule $f_0\in\mathcal F$ and set
\begin{equation}\label{eq:norm_rank}
(\alpha_{f_0},\beta_{f_0})=(0,0).
\end{equation}
Let $\vartheta\in\mathbb R^{K}$ collect the remaining free parameters, with
\[
K=(|\mathcal F|-1)\,(1+\dim(\psi)).
\]

\paragraph{Decisive-side indicators and the $H$ matrix.}
For each rule $f$ and menu $A$, define the decisive-side indicators
\[
\kappa_f^L(A):=\Ind\{f\in\mathcal I_L(A)\},\qquad
\kappa_f^R(A):=\Ind\{f\in\mathcal I_R(A)\},
\]
and write the odds ratio
\[
r(A):=\frac{p(A)}{1-p(A)}.
\]
For any $x\in\mathbb R^d$, define the positive weight vector
\[
\omega_f(x):=\exp(\alpha_f+\beta_f^\top x).
\]
Cross-multiplying the odds form of the model yields, for each menu $A$ with $\psi(A)=x$,
\begin{equation}\label{eq:linear_restriction}
\sum_{f\in\mathcal F}\big(\kappa_f^L(A)-r(A)\kappa_f^R(A)\big)\,\omega_f(x)=0,
\end{equation}
i.e.\ $h(A)^\top \omega(x)=0$, where
\[
h_f(A):=\kappa_f^L(A)-r(A)\kappa_f^R(A).
\]

\paragraph{Identification in two layers.}
Equation \eqref{eq:linear_restriction} reveals the first identifying object: at a fixed feature value $x$,
every menu with $\psi(A)=x$ shares the same latent weight vector $\omega(x)$, but may display a different
decisive-side pattern. These distinct decisive-side patterns generate linear restrictions on the same unknown
vector $\omega(x)$. Once $\omega(x)$ is identified at sufficiently many feature values, the affine structure
of the gate recovers $(\alpha_f,\beta_f)$ rule by rule.

\begin{theorem}[Fixed-feature identification of rule weights]\label{thm:fixed_x_id}
Fix $x\in\mathbb R^d$ and define
\[
\mathcal A_2(x):=\{A\in\mathcal A_2:\psi(A)=x\}.
\]
For each $A\in\mathcal A_2(x)$, let
\[
h_f(A):=\kappa_f^L(A)-r(A)\kappa_f^R(A),
\qquad
r(A):=\frac{p(A)}{1-p(A)}.
\]
Define the cellwise identified set
\[
\Omega(x):=\Big\{\omega\in\mathbb R_{++}^{|\mathcal F|}: h(A)^\top\omega=0
\ \ \forall A\in\mathcal A_2(x)\Big\}.
\]
Assume $\Omega(x)\neq\varnothing$ (this is guaranteed under correct specification: the
true weights $\omega_f(x)=\exp(\alpha_f+\beta_f^\top x)$ are strictly positive by
construction, so $\omega(x)\in\Omega(x)$). Then the following are equivalent:
\begin{enumerate}[label=(\roman*)]
\item $\Omega(x)$ is a single positive ray, i.e.\ there exists $\omega^\star(x)\in\mathbb R_{++}^{|\mathcal F|}$ such that
\[
\Omega(x)=\{c\,\omega^\star(x): c>0\};
\]
\item the row span of $\{h(A):A\in\mathcal A_2(x)\}$ has dimension $|\mathcal F|-1$.
\end{enumerate}
Equivalently, if $H(x)$ is any stacked matrix whose rows span the same subspace as
$\{h(A):A\in\mathcal A_2(x)\}$, then
\[
\Omega(x)\text{ is a single positive ray}
\quad\Longleftrightarrow\quad
\rank H(x)=|\mathcal F|-1.
\]
Under the normalization $\omega_{f_0}(x)=1$, this is equivalent to point identification of $\omega(x)$.
\end{theorem}

\begin{proof}[Proof sketch]
The model implies $h(A)^\top\omega(x)=0$ for every $A\in\mathcal A_2(x)$, so $\Omega(x)$ is the set of strictly positive vectors in the null space of the stacked matrix $H(x)$. If $\rank H(x)=|\mathcal F|-1$, the null space is one-dimensional and $\Omega(x)$ is a single positive ray; the converse follows because, if the null space has dimension at least two and already contains one strictly positive vector (namely the true $\omega(x)$), then a sufficiently small perturbation along an independent null-space direction yields a second, non-proportional positive vector that still lies in the positive orthant. See Appendix~\ref{app:proof_fixed_x} for the full proof.
\end{proof}

The key identifying force in Theorem~\ref{thm:fixed_x_id} is \emph{rule switching at fixed features}: menus sharing the same $\psi(A)=x$ share the same latent weight vector $\omega(x)$, but may differ in which rules recommend which option. Each such menu contributes a linear restriction $h(A)^\top\omega(x)=0$, and identification occurs when these restrictions span a codimension-one subspace of $\mathbb{R}^{|\mathcal F|}$.

\paragraph{From fixed-feature weights to global identification of the affine gate.}
Theorem~\ref{thm:fixed_x_id} identifies the first-stage object $\omega(x)$ feature value by feature
value. The next result shows that once these weight vectors are recovered at sufficiently many
affinely independent feature values, the affine gate parameters themselves are globally identified.

\begin{corollary}[Global identification from fixed-feature weights]\label{cor:global_id_from_fixedx}
Consider the CA-RRM with affine indices
\[
s_f(A)=\alpha_f+\beta_f^\top \psi(A),\qquad f\in\mathcal F,
\]
and impose the normalization
\[
(\alpha_{f_0},\beta_{f_0})=(0,0)
\]
for some baseline rule $f_0\in\mathcal F$. Let $d=\dim(\psi)$.
Suppose there exist $(d+1)$ feature values $x^{(0)},\ldots,x^{(d)}$ such that:

\begin{enumerate}[label=(G\arabic*)]
\item \textbf{Cellwise identification:} For each $k\in\{0,\ldots,d\}$, there exist
menus $A_{k1},\ldots,A_{kM_k}$ with $\psi(A_{km})=x^{(k)}$ for all $m$, such that the associated
matrix
\[
H^{(k)}:=\begin{pmatrix}
h(A_{k1})^\top\\
\vdots\\
h(A_{kM_k})^\top
\end{pmatrix}
\]
has rank $|\mathcal F|-1$.

\item \textbf{Affine richness of feature values:} The matrix
\[
X:=\begin{pmatrix}
1 & (x^{(0)})^\top\\
\vdots & \vdots\\
1 & (x^{(d)})^\top
\end{pmatrix}
\]
has full rank $d+1$.
\end{enumerate}

Then the full parameter vector $\{(\alpha_f,\beta_f)\}_{f\in\mathcal F}$ is globally identified
from the population objects $\{p(A),\kappa^L(A),\kappa^R(A),\psi(A)\}$, up to the normalization.
\end{corollary}

\begin{proof}[Proof sketch]
By Theorem~\ref{thm:fixed_x_id}, (G1) identifies each $\omega(x^{(k)})$ up to scale; the normalization $\omega_{f_0}\equiv 1$ pins the scale. For each rule $f$, stacking $\log\omega_f(x^{(k)})=\alpha_f+\beta_f^\top x^{(k)}$ over $k=0,\ldots,d$ gives a linear system with coefficient matrix $X$, which is invertible by~(G2). See Appendix~\ref{app:proof_global_id} for the detailed version.
\end{proof}

The corollary separates global identification into two conceptually distinct requirements.
Condition (G1) is a within-feature condition: it uses rule switching across menus that share the
same feature value to identify the relative rule weights at that feature value. Condition (G2) is
an across-feature condition: it uses affine variation in the feature values themselves to map the
identified weight vectors back into the gate coefficients. Thus, the identification argument has a
clear two-step structure: first recover $\omega(x)$ locally in feature space, then recover the
affine law of motion of those weights globally.

\paragraph{Effective dimension and baseline specification.}
Corollary~\ref{cor:global_id_from_fixedx} is stated for a generic feature map~$\psi$.
In our baseline specification (Section~\ref{subsec:gate_features}), the gate operates on a
$12$-dimensional vector of interpretable summary statistics $z(A)\in\mathbb R^{12}$, so the
corollary applies with $\psi=z$. When this feature vector has structural redundancies
(e.g., exact linear dependencies among coordinates), the corollary applies with $d$ replaced by
the effective dimension
\[
d_{\mathrm{eff}}:=\rank([\mathbf{1},z])-1,
\]
and identification holds for the gate parameters projected onto the effective feature subspace
(i.e., up to the null space of $[\mathbf{1},z]$). In our application, $d_{\mathrm{eff}}=11$
(the expected-value gap is a linear combination of the two level expected values). The neural
benchmark models use a separate, higher-dimensional raw menu encoding
$\psi(A)\in\mathbb R^{4M}$ as their input (Section~\ref{subsec:encoding}); the identification
conditions and diagnostics below do not target that encoding.

\subsubsection*{Diagnostics: empirically assessing identification strength}\label{subsec:rank_diagnostics}

Theorem~\ref{thm:fixed_x_id} and Corollary~\ref{cor:global_id_from_fixedx}
reduce global identification to two verifiable ingredients on the observed support:
cellwise rank sufficiency and affine richness of the identifying feature values.
We report three diagnostics.

D1 reports the fraction of two-sided menus (ensuring the log-odds representation is empirically meaningful). D2 reports per-rule decisiveness frequency and side-variation conditional on decisiveness; rules that take both sides contribute sign-varying rows to the $H^{(k)}$ matrices, supporting (G1). D3 verifies (G1) and (G2) directly: menus are grouped into cells sharing the same (or nearby) value of the gate features $z(A)$, and $\rank(\hat H^{(k)})$ is computed for each cell via SVD; (G1) is supported if at least $d_{\mathrm{eff}}+1$ cells achieve rank $|\mathcal F|-1$, and (G2) is supported if the centroids of those cells are affinely independent. We report the empirical results of these diagnostics in Section~\ref{subsec:rank_empirics}.

\section{Rule library}\label{subsec:library_details}

The library is deliberately small and simple: twelve parameter-free heuristics, each stated as a single line of conduct and implemented deterministically menu-by-menu. Rather than postulating a global valuation model and fitting preference parameters, we treat behavior as switching among a few coarse procedures.

Table~\ref{tab:rules_summary} gives the complete library. Each rule maps a binary lottery menu to a direct recommendation --- ``choose Left,'' ``choose Right,'' or ``abstain'' --- depending on the comparison it performs. The rule is \emph{active} ($A_f(A)=1$) when it delivers a strict ranking and \emph{inactive} otherwise.

\begin{table}[t]
\centering
\small
\caption{Rule library: twelve deterministic decision rules.\label{tab:rules_summary}}
\begin{tabular}{llll}
\toprule
\textbf{Family} & \textbf{Rule} & \textbf{Recommendation} & \textbf{Active when} \\
\midrule
Extremum & \textsc{MMn} & Choose higher $\min$ outcome & $x_{\min}(L^1)\neq x_{\min}(L^2)$ \\
& \textsc{MMx} & Choose higher $\max$ outcome & $x_{\max}(L^1)\neq x_{\max}(L^2)$ \\
& \textsc{MMa} & Choose higher midrange & midranges differ \\
& \textsc{MAP} & Choose higher modal outcome & modes differ \\
\addlinespace
Salience & \textsc{SAL} & Choose per most salient extreme comparison & comparison is strict \\
& \textsc{SAL2} & Choose per second-most salient comparison & unique second rank \\
\addlinespace
Regret & \textsc{REG} & Choose lower negated-regret distribution (FSD) & FSD ranking is strict \\
& \textsc{REGmed} & Choose lower median regret & median regrets differ \\
\addlinespace
Disappt. & \textsc{DIS} & Avoid largest downside contrast below mode & contrasts differ \\
& \textsc{DISmed} & Avoid second-largest downside contrast & contrasts differ \\
\addlinespace
Attention & \textsc{A1} & Always choose $L^1$ & always \\
& \textsc{A2} & Always choose $L^2$ & always \\
\bottomrule
\end{tabular}
\end{table}

\subsubsection*{Generating the indicators via perceived-lottery transformations}

To define activity and recommendation uniformly across all twelve rules, we use a common implementation device from \citet{seror2026}. Let $X$ denote the set of finite-support lotteries over monetary outcomes. Each rule $f$ maps a binary menu to a \emph{perceived binary menu}:
\[
f:\ X\times X\to X\times X,\qquad f(L^1,L^2)=\big(\pi_f(L^1;L^2),\,\pi_f(L^2;L^1)\big),
\]
where $\pi_f(L^i;L^j)\in X$ is a finite-support lottery. The activity and recommendation indicators are then determined by comparing perceived lotteries via First-Order Stochastic Dominance (FSD):
\begin{align}
A_f(A) &= \Ind\Big\{\pi_f(L^1;L^2) >_{\mathrm{FSD}} \pi_f(L^2;L^1)\ \ \text{or}\ \ \pi_f(L^2;L^1) >_{\mathrm{FSD}} \pi_f(L^1;L^2)\Big\},\notag\\
L_f(A) &= \Ind\Big\{\pi_f(L^1;L^2) >_{\mathrm{FSD}} \pi_f(L^2;L^1)\Big\}.\notag
\end{align}
For ten of the twelve rules, the perceived lotteries are degenerate (sure outcomes), so the FSD comparison reduces to comparing two numbers. For the regret rules, the perceived lotteries are genuine distributions. This construction is an implementation device - the model and identification theory require only the indicators $(A_f,L_f)$ - but it provides a uniform formal language for heterogeneous heuristics. The formal perceived-lottery map $\pi_f$ for each rule, along with the notation and auxiliary definitions needed to state them, is given in Appendix~\ref{app:perceived_lottery_specs}.

\subsubsection*{Discussion}

The perceived-lottery construction is a uniform implementation device: it generates the indicators $(A_f, L_f)$ for each rule via a single comparison criterion (FSD), allowing heterogeneous heuristics - from simple outcome comparisons to distributional regret evaluations - to be defined within a common formal language. The CA-RRM and all identification results are stated in terms of $(A_f, L_f)$ and do not depend on this construction.

The library spans five broad mechanisms emphasized in behavioral and bounded-rationality work - outcome simplification, salience-driven attention, regret, disappointment, and limited consideration - implemented here in a deliberately coarse, parameter-free form. The gate does not estimate a valuation model; it learns when each line of conduct best predicts behavior. Several rules are variants within a common mechanism (e.g.\ \textsc{SAL} vs.\ \textsc{SAL2}, \textsc{REG} vs.\ \textsc{REGmed}). When two variants are active at the same menu, the conditional-on-activity weights \eqref{eq:qtilde_def} yield a direct within-family split. A natural alternative would be a single ``aggregate'' rule that combines information across all ranks, but such a rule requires specifying how ranks are weighted --- absent preference parameters there is no canonical choice. We instead use discrete variants (top-ranked vs.\ second-ranked comparison) that probe sensitivity to the ranking without introducing an explicit aggregation scheme, stopping at rank two to keep the library small.

\subsection{Rule-level Diagnostics}\label{sec:diagnostics}

The rule-gating predictor is transparent at the menu level: by \eqref{eq:mixture_form}, each predicted choice
probability is a convex combination of deterministic rule recommendations, with weights
$\tilde q_f(A;\hat\theta)$ that can be interpreted as \emph{conditional-on-activity responsibilities}.
This section introduces summary diagnostics that turn these menu-by-menu decompositions into interpretable
rule-level statements. The goal is twofold: (i) quantify which procedures the fitted gate actually uses
on average, and (ii) distinguish rules that are essential for fit from rules that mainly act as substitutes.

\subsubsection*{Effective responsibility weights and concentration}\label{subsec:usage_HHI}

Because the CA-RRM conditions on activity, the economically meaningful notion of ``rule usage'' is the gate mass
\emph{among decisive rules}. For each menu $A_t$ and each rule $f$, define the effective responsibility
\begin{equation}\label{eq:qtilde_i_def}
\tilde q_{tf}
:=\tilde q_f(A_t;\hat\theta)
=\frac{q_f(A_t;\hat\theta)\,A_f(A_t)}{\sum_{g\in\mathcal F} q_g(A_t;\hat\theta)\,A_g(A_t)},
\end{equation}
whenever the denominator is positive.\footnote{In our implementation, the two attention rules \textsc{A1} and \textsc{A2}
are always included and are decisive for essentially all menus, so the denominator is typically well away from zero.}
The \emph{effective responsibility weight} of rule $f$ is the average responsibility across menus,
\begin{equation}\label{eq:wf_def}
w_f:=\frac{1}{T}\sum_{t=1}^T \tilde q_{tf},\qquad \sum_{f\in\mathcal F} w_f=1.
\end{equation}
Thus $w_f$ measures the average share of decision mass attributed to rule $f$, \emph{conditional on the rule being
eligible to speak} (i.e., active) at each menu.

\subsubsection*{Rule Concentration}

To summarize whether the fitted mixture concentrates on a few procedures or is diffuse across many, we report the
Herfindahl concentration of responsibility weights,
\begin{equation}\label{eq:HHI}
\mathrm{HHI}(F):=\sum_{f\in\mathcal F} w_f^2.
\end{equation}
Higher $\mathrm{HHI}(F)$ indicates that a small subset of rules accounts for most effective responsibility, whereas lower
$\mathrm{HHI}(F)$ indicates a more even allocation across rules. This statistic is informative because the
rule-gating predictor is interpretable only to the extent that its effective responsibilities are \emph{parsimonious}:
a high concentration means that most menus are explained by a small repertoire of named procedures, while a low
concentration suggests that predictive success relies on a broad set of substitutes. We also report the \emph{effective number of rules}
\[
N_{\mathrm{eff}}(F):=\frac{1}{\mathrm{HHI}(F)}\in[1,|\mathcal F|],
\]
which can be read as the number of equally-weighted rules that would generate the same concentration. For example,
$N_{\mathrm{eff}}=3$ means that the fitted model is about as concentrated as one that splits effective responsibility evenly
across three procedures. We use $\mathrm{HHI}(F)$ (and $N_{\mathrm{eff}}$) as a compact summary of whether the fitted
rule-gating model achieves accuracy through a few unifying mechanisms or through a diffuse patchwork of rules.

\subsubsection*{Predictive ablation index \texorpdfstring{$\phi(f)$}{phi(f)}}\label{subsec:phi}

Responsibility weights are descriptive, but they do not by themselves tell whether a rule is \emph{necessary} for predictive
accuracy: a rule can carry substantial mass yet be redundant because other rules typically make the same
recommendation, and conversely a rarely-used rule can matter sharply in a small region of menu space.
To quantify importance for fit, we use a refit-and-compare ablation diagnostic.

For each rule $f\in\mathcal F\setminus\{\textsc{A1},\textsc{A2}\}$, let $\mathrm{MSE}^{-f}$ denote the cross-validated
test MSE of rule-gating when $f$ is removed from the library and the gate is retrained on the remaining rules.
Let $\mathrm{MSE}^{\mathrm{full}}$ be the corresponding MSE for the full library. Define
\begin{equation}\label{eq:phi}
\phi(f):=\frac{\mathrm{MSE}^{-f}-\mathrm{MSE}^{\mathrm{full}}}{\mathrm{MSE}^{\mathrm{full}}}.
\end{equation}
Thus $\phi(f)>0$ means that removing $f$ worsens out-of-sample accuracy, and larger values indicate a more
prediction-critical procedure.\footnote{We exclude \textsc{A1}/\textsc{A2} from ablation because they serve as
stability anchors by ensuring decisiveness at essentially every menu; dropping them would confound behavioral
importance with numerical degeneracy of the CA-RRM normalization.}

\subsubsection*{Concentration impact \texorpdfstring{$\sigma_N(f)$}{sigma(f)}}\label{subsec:sigma}

Finally, we ask whether a rule acts as a \emph{unifier} (its presence makes effective responsibility concentrate on a small
repertoire of procedures) or instead as a \emph{substitute} that helps spread responsibility across the library.
For each $f\in\mathcal F\setminus\{\textsc{A1},\textsc{A2}\}$, we refit the gate after removing $f$ and recompute the
effective responsibility weights and concentration. Let $\mathrm{HHI}^{-f}$ denote the resulting concentration, and
let $N_{\mathrm{eff}}^{-f}:=1/\mathrm{HHI}^{-f}$ be the associated effective number of rules (recall $N_{\mathrm{eff}}(F)=1/\mathrm{HHI}(F)$ from above).
We summarize the \emph{concentration impact} of rule $f$ by the relative change in the effective number of rules:
\begin{equation}\label{eq:sigma_neff}
\sigma_{N}(f):=\frac{N_{\mathrm{eff}}^{-f}-N_{\mathrm{eff}}(F)}{N_{\mathrm{eff}}(F)}.
\end{equation}
Thus $\sigma_{N}(f)>0$ means that removing $f$ increases the effective number of rules (the fitted mixture becomes
more diffuse). Conversely, $\sigma_{N}(f)<0$ means that removing $f$ decreases the effective number of rules (the mixture becomes
more concentrated), so $f$ primarily acts as a substitute/diversifier: when it is unavailable, the gate reallocates
mass more sharply onto a smaller set of remaining procedures. This refit-based diagnostic complements the ablation index $\phi(f)$: $\phi(f)$ measures the \emph{predictive}
importance of $f$, while $\sigma_{N}(f)$ measures its role in shaping the \emph{structure} of the fitted mixture.
Importantly, $\sigma_{N}(f)$ is computed after a full refit and therefore reflects equilibrium reallocation of
responsibility.

The pair $(\phi(f),\sigma_N(f))$ separates distinct roles.
A large $\phi(f)$ indicates that the procedure improves predictive fit and is not easily replaced by the rest of the
library; the sign and magnitude of $\sigma_N(f)$ indicate whether that procedure also shapes the overall mixture
structure (concentrating or diversifying effective responsibility). In Section~\ref{sec:application} we report $(w_f,\phi(f),\sigma_N(f))$
side-by-side to identify rules that are heavily used, rules that are essential for accuracy, and rules that mainly act
as substitutes or robustness variants.

\section{Estimation and inference}\label{subsec:estimation_inference}

The identification results of Section~\ref{sec:identification_switching} are constructive: they imply a two-step estimator that recovers the gate parameters and supports inference on rule-level functionals such as responsibility weights and ablation indices.

The first step groups menus into cells sharing the same (or nearby) gate-feature value $\psi(A)=x^{(k)}$. Within each cell, the linear restrictions $h(A)^\top\omega(x^{(k)})=0$ generated by rule switching (Theorem~\ref{thm:fixed_x_id}) yield a constrained least-squares estimate of the normalized weight vector $\hat\omega(x^{(k)})$. The second step regresses the estimated log-weights $\log\hat\omega_f(x^{(k)})$ on the cell centroids to recover the affine gate coefficients $(\hat\alpha_f,\hat\beta_f)$.

Under standard regularity conditions - interior choice probabilities, cellwise rank sufficiency, affine richness of the identifying feature values, and stable sampling shares - the two-step estimator is consistent and asymptotically normal:

\begin{theorem}[Consistency and asymptotic normality of the two-step estimator]
\label{thm:two_step}
Let $N:=\min_m n_m\to\infty$ with $n_m/N\to\pi_m\in(0,\infty)$. Under the regularity conditions stated in Appendix~\ref{app:two_step_details} (Assumption~\ref{ass:estimation}):
\begin{enumerate}[label=(\roman*)]
\item The cellwise weight estimates $\hat\omega(x^{(k)})$ are consistent and $\sqrt{N}$-asymptotically normal.
\item The affine gate coefficients $\hat\gamma_f=(\hat\alpha_f,\hat\beta_f)$ are consistent and $\sqrt{N}$-asymptotically normal, with a sandwich covariance matrix reflecting the two-step structure.
\item Any continuously differentiable functional of the gate parameters (e.g., average responsibility weights $w_f$) inherits consistency and asymptotic normality by the delta method.
\end{enumerate}
\end{theorem}

The estimator construction, regularity conditions, proof, and covariance formulas are given in Appendix~\ref{app:two_step_details}. When the number of cells exceeds $d+1$ (the dimension of the augmented feature vector), the affine gate restriction is overidentified, and a standard minimum-distance $J$-test can be used to assess the specification; we report this diagnostic in Section~\ref{subsec:two_step_results}.

This two-step procedure targets identification-consistent parameter recovery and supports formal inference on the decomposition objects. For out-of-sample prediction and model comparison, however, we estimate the gate by direct MSE minimization (Section~\ref{subsec:est_rg}), which avoids the binning of the feature space into discrete cells that the two-step procedure requires. Both estimators target the same population decomposition when the gate specification is correct; the two-step estimator provides a built-in diagnostic for this restriction.

\section{Data and empirical design}\label{sec:design}

This section introduces the data and evaluation protocol.
We describe the dataset, define the prediction target and evaluation metric,
explain the gate-feature specification on which the identification conditions of Section~\ref{sec:identification_switching}
are verified, and detail the cross-validation protocol used for predictive comparison.
Benchmark models are described at the end of the section.

\subsection{Dataset and target variable}\label{subsec:data}

We use the \texttt{choices13k} dataset introduced by \citet{peterson2021}, containing 13,006 binary risky-choice problems. Each problem
corresponds to a menu $A=\{L^1,L^2\}$ and is repeated across many observations, allowing the computation of a menu-level choice frequency.
We work with the filtered dataset that (i) removes ambiguous problems (masked probabilities)
and (ii) keeps only problems for which feedback was provided, yielding $T=9{,}831$ menus. The feedback restriction follows the baseline protocol of \citet{peterson2021} and has a practical justification: feedback problems have more repetitions per menu, producing more precise estimates of the choice frequency $\hat p(A)$ that serves as our prediction target. Because our model is static (it maps menu features to a predicted choice probability without using feedback information), the feedback condition affects only the precision of the dependent variable, not the model's information set.\footnote{Feedback can in principle induce learning effects. If such dynamics are present, our static model captures the average behavior under the feedback regime.} For each menu $A$, the dataset reports the fraction of times the left option is chosen. We denote this empirical frequency by
$\hat p(A)$ and treat it as the target to be predicted.

\subsection{Evaluation metric}\label{subsec:loss}

Our main performance metric is the mean squared error (MSE):
\begin{equation}\label{eq:mse}
\mathrm{MSE}(g)=\frac{1}{T}\sum_{t=1}^T\Big(g(A_t)-\hat p(A_t)\Big)^2,
\end{equation}
for $\{A_t\}_{t\in \{1,\dots,T\}}\subseteq \mathcal A_{\text{full}}$. The MSE is a scoring rule: it rewards probabilistic accuracy and penalizes both systematic bias and overconfidence. We focus on MSE because
it is standard and is the primary metric reported by \citet{peterson2021}.

Because the number of trials $n(A_t)$ behind each menu-level frequency $\hat p(A_t)$ varies across menus, we also report a trial-weighted variant:
\begin{equation}\label{eq:mse_w}
\mathrm{MSE}_w(g)=\frac{\sum_{t=1}^T n(A_t)\Big(g(A_t)-\hat p(A_t)\Big)^2}{\sum_{t=1}^T n(A_t)}.
\end{equation}
This metric downweights menus estimated from few trials, where $\hat p(A_t)$ is noisier, and gives a more precision-aware summary of predictive accuracy. We use the unweighted MSE as the primary metric for comparability with \citet{peterson2021} and report $\mathrm{MSE}_w$ as a robustness check.

We also conduct paired split-level comparisons (Appendix Table~\ref{tab:paired_comparisons}). Because Choices13k menus typically aggregate many repetitions (median above 25), we treat $\hat p(A)$ as the prediction target without modeling the sampling process; $\mathrm{MSE}_w$ provides a precision-aware robustness check.

\subsection{Menu encoding}\label{subsec:encoding}

All models take as input a fixed-length numerical encoding of the two lotteries in the menu. To align with \citet{peterson2021}, each lottery is represented
by its list of outcomes and probabilities after sorting outcomes in increasing order, truncating to a maximum support size $M$ and padding with zeros.
We set $M=10$ throughout.

Formally, for each lottery $L$ with support points $\{(x_k,\pi_k)\}_{k=1}^{K}$, order outcomes so that $x_1\le \cdots \le x_K$ and define
\[
x_{1:M}(L)=(x_1,\ldots,x_{\min\{K,M\}},0,\ldots,0),\qquad
\pi_{1:M}(L)=(\pi_1,\ldots,\pi_{\min\{K,M\}},0,\ldots,0).
\]
The menu encoding is the concatenation
\begin{equation}\label{eq:psi}
\psi(A)=\Big(x_{1:M}(L^1),\pi_{1:M}(L^1),x_{1:M}(L^2),\pi_{1:M}(L^2)\Big)\in\mathbb R^{4M}.
\end{equation}
To stabilize optimization across models, we rescale outcomes by the maximum absolute payoff observed in the filtered sample; the same scaling
is applied for all models and all folds. This rescaling enters only the menu encoding $\psi(A)$ (and hence the gate features $z(A)$) used as inputs to differentiable models; it does not affect the rule indicators, which are computed on raw outcomes before any rescaling. The scaling is a fixed, data-wide preprocessing step applied identically to all models and folds.

\subsection{Cross-validation protocol}\label{subsec:cv}

We compare models using the out-of-sample evaluation structure of \citet{peterson2021}. The unit of observation for prediction is the menu, not the individual trial.
We repeatedly split the set of menus into training and test sets, fit model parameters on the training set, and evaluate MSE on the held-out test set.

\paragraph{Splits.}
We generate $K=50$ independent random splits of the $N$ menus into $90\%$ training and $10\%$ test, as in \citet{peterson2021}. Each model is trained from scratch on each split.

\paragraph{Learning-rate selection.}
Each differentiable model is trained by gradient descent (Adam). Following \citet{peterson2021}, we select
the learning rate from a small grid. To avoid selecting the learning rate on the same data used to evaluate generalization (as in \citealp{peterson2021}, who select on test MSE), we use a two-pass protocol within each
split. In Pass~A (validation), we further partition the training set into an 80\% sub-training set and a 20\% validation set, train the model on the sub-training set for each candidate learning rate, and record the validation MSE. In Pass~B (test), we retrain the model on the \emph{full} training set using each candidate learning rate and evaluate MSE on the held-out test set. The reported performance of a model is the average test MSE (from Pass~B) across the 50 splits at the learning rate that achieves the lowest average validation MSE (from Pass~A). This two-pass validation protocol applies to all models \emph{except} the Mixture of Theories (MOT). MOT relies on the HURD library's internal training/validation logic, which does not expose hooks for overriding the train/validation partition or for extracting per-menu predictions at the validation stage. We therefore select MOT's learning rate by mean \emph{test} MSE across splits, following \citet{peterson2021}. This asymmetry is conservative: it gives MOT the benefit of test-set-informed learning-rate selection, which, if anything, biases the comparison in MOT's favor.

\paragraph{Implementation details.}
We adopt the same network sizes and training horizons used by \citet{peterson2021} (and their Supplementary Materials) for Neural EU/PT/CPT, Value-Based,
Context-Dependent, and MOT. We keep all such hyperparameters fixed across splits and do not tune architectures. The only tuned hyperparameter is the learning rate.

\subsection{Predictive implementation}\label{subsec:est_rg}

For the purpose of out-of-sample prediction and model comparison, we estimate the gate parameters by direct MSE minimization. This predictive estimator differs from the two-step econometric procedure of Section~\ref{subsec:estimation_inference}: the latter targets the identification-driven parameter recovery and supports inference for rule-level functionals, while the former is designed purely for predictive accuracy and benchmark evaluation. The rule-gating model differs from the benchmark neural predictors in a key way: the \emph{rules are deterministic and parameter-free}.
Estimation concerns only the \emph{gate}, i.e.\ the mapping from menu features $z(A)$ to rule propensities $q(A)\in\Delta^{F-1}$ (see Section~\ref{subsec:gate_features} for the definition of $z(A)$).

\paragraph{Precomputation: rule activity and recommendations.}
For each menu $A_t$ and each rule $f\in\mathcal F$, we compute once-and-for-all the indicators
$A_f(A_t)$ and $L_f(A_t)$ defined in \eqref{eq:Af_def}--\eqref{eq:Lf_def}. These are deterministic functions of $(A_t,f)$.
Precomputing them is important: it makes training fast and ensures that the gate is the only learned component.

\paragraph{Gate estimation.}
On a given training split, we estimate the gate parameters $\theta$ by minimizing the training MSE between the predicted choice probability
$g_{\mathrm{RG}}(A;\theta)$ and the target $\hat p(A)$, where
\[
g_{\mathrm{RG}}(A;\theta)=\frac{\sum_{f\in\mathcal F} q_f(A;\theta)\,L_f(A)}{\max\{\sum_{f\in\mathcal F} q_f(A;\theta)\,A_f(A),\,m_{\min}\}}.
\]
The gate weights $q_f(A;\theta)$ are given by the softmax specification \eqref{eq:softmax_gate}. We optimize the smooth objective using Adam with gradient clipping.
This training step is conceptually simple: it is standard least-squares fitting, except that predictions are constructed as a weighted average of deterministic rule
recommendations, conditional on rule activity.

\subsection{Gate features}\label{subsec:gate_features}

The gate operates on a $12$-dimensional vector of economically motivated summary statistics $z(A)\in\mathbb{R}^{12}$: six pairwise gaps between the two lotteries (EV, max, min, variance, mode, skewness) and six per-lottery or menu-level features (EV and SD of each lottery, the largest absolute outcome, and the support-size gap). These features capture the contrasts that the rules operate on (min/max gaps relate to \textsc{MMn}/\textsc{MMx}, the mode gap to \textsc{MAP}, etc.) while allowing the gate to condition on menu scale and complexity.

Using $z(A)$ rather than the raw $40$-dimensional encoding $\psi(A)$ is conceptually appropriate: the gate should capture systematic variation in rule usage with broad menu characteristics, not index menus individually. Coarsening also yields repeated support points needed for the identification conditions (G1)--(G2). In our application, $z(A)$ matches or improves predictive accuracy relative to the raw encoding (Appendix~\ref{app:raw_gate}). The gate uses the softmax architecture of Equation~\eqref{eq:softmax_gate} with $z(A)$ in place of $\psi(A)$, giving $(|\mathcal{F}|-1)\times(1+d_z)$ free parameters.

\subsection{Benchmark models}\label{subsec:benchmarks}

We benchmark rule-gating against the differentiable model classes studied by \citet{peterson2021}. Most benchmarks assign each gamble a value $V(\cdot)$ and map values into a choice probability via a logit link $g(A;\theta)=\exp(\eta V(L^1))/[\exp(\eta V(L^1))+\exp(\eta V(L^2))]$, where $\eta$ is a free inverse-noise parameter. Neural EU, PT, and CPT learn the utility function $u(\cdot)$ (and, for PT/CPT, the probability-weighting function $\pi(\cdot)$) as small neural networks while retaining the respective aggregation structure. Value-Based relaxes the aggregation, learning $V(L)=f(x^L,p^L)$ as a generic neural network. Context-Dependent is the most flexible: a neural network directly maps the full description of both gambles to a choice probability without separate valuation. Finally, the Mixture of Theories (MOT) is a structured context-dependent benchmark that softly selects among candidate utility and probability-weighting functions via gating networks \citep[Supplementary Materials]{peterson2021}. All architectures and training horizons follow \citet{peterson2021}.

\section{Application: Procedural Heterogeneity in Risky Choice}\label{sec:application}

This section estimates procedural heterogeneity in risky choice using the \texttt{choices13k} dataset.
We first verify the identification conditions on the data and establish out-of-sample predictive performance.
We then present the estimated behavioral decomposition - responsibility weights, ablation indices,
and concentration measures - followed by an econometric cross-check via the two-step estimator.
Finally, we report comparative statics in rule responsibility along menu primitives
and portability to an independent dataset.
The responsibility weights $w_f$ and context-dependent patterns reported below are \emph{population-level} objects: they describe how the fitted mixture varies across menus in a dataset aggregating choices across many individuals. Without individual-level panel data, within-individual and between-individual components cannot be separated; mechanism-level statements should be read as describing population-level reallocation.

\subsection{Identification diagnostics}\label{subsec:rank_empirics}

We evaluate diagnostics D1--D3 on the full Choices13k dataset.

\paragraph{D1.} $\dOneCount$ of $\dOneTotal$ menus ($\dOneFrac$\%) are two-sided,
confirming that the log-odds representation \eqref{eq:delta_rank} is well-defined on virtually the entire sample.

\paragraph{D2.}
Table~\ref{tab:coverage} reports per-rule coverage and side-variation.
All ten non-attention rules switch sides, and all but one (\textsc{REG}, 30\%) are decisive on at least 41\% of menus,
contributing sign-varying rows to the $H^{(k)}$ matrices that drive condition~(G1).
The attention rules (\textsc{A1}, \textsc{A2}) are each confined to a single side by construction.

\begin{table}[t]
\centering
\caption{\label{tab:coverage}Rule-level coverage and side-variation (D2). For each rule $f$, $N_{\text{act}}$ is the number of menus where $f$ is decisive, $\Pr(L\mid\text{act})$ and $\Pr(R\mid\text{act})$ are the conditional left/right recommendation frequencies. Rules that take both sides contribute sign-varying rows to $H^{(k)}$, supporting condition (G1).}
\begin{tabular}{lrrccc}
\toprule
Rule & $N_{\text{act}}$ & $\Pr(\text{act})$ & $\Pr(L\mid\text{act})$ & $\Pr(R\mid\text{act})$ & Switches \\
\midrule
\textsc{MMn} & 9,584 & 0.975 & 0.294 & 0.706 & $\checkmark$ \\
\textsc{MMa} & 9,668 & 0.983 & 0.580 & 0.420 & $\checkmark$ \\
\textsc{MMx} & 9,688 & 0.985 & 0.739 & 0.261 & $\checkmark$ \\
\textsc{MAP} & 9,512 & 0.968 & 0.410 & 0.590 & $\checkmark$ \\
\textsc{SAL} & 9,831 & 1.000 & 0.485 & 0.515 & $\checkmark$ \\
\textsc{SAL2} & 4,057 & 0.413 & 0.509 & 0.491 & $\checkmark$ \\
\textsc{REG} & 2,957 & 0.301 & 0.543 & 0.457 & $\checkmark$ \\
\textsc{REGmed} & 8,661 & 0.881 & 0.446 & 0.554 & $\checkmark$ \\
\textsc{DIS} & 4,773 & 0.486 & 0.322 & 0.678 & $\checkmark$ \\
\textsc{DISmed} & 4,773 & 0.486 & 0.336 & 0.664 & $\checkmark$ \\
\textsc{A1} & 9,831 & 1.000 & 1.000 & 0.000 &  \\
\textsc{A2} & 9,831 & 1.000 & 0.000 & 1.000 &  \\
\bottomrule
\end{tabular}
\end{table}

\paragraph{D3.}
Table~\ref{tab:global_id} evaluates the rank conditions (G1) and (G2) of Corollary~\ref{cor:global_id_from_fixedx} on the baseline gate features $z(A)\in\mathbb{R}^{12}$.
In the Choices13k data, the $12$-dimensional features are essentially unique across menus, so exact-equality cells are singletons. We supplement exact equality with $k$-means clustering on $z(A)$, grouping menus into $\gidNclusters$ cells.\footnote{The number of clusters ($\gidNclusters$) is chosen so that each cell contains at least $|\mathcal F|-1$ menus while keeping within-cell feature variation modest.} Of these, $\gidCellsRankOk$ cells achieve $\mathrm{rank}(H^{(k)})=|\mathcal F|-1$, well above the $d_{\mathrm{eff}}+1=\gidDeffPlusOne$ needed for (G1). Among the G1-passing cells, the augmented centroid matrix has full affine rank, satisfying (G2). Both conditions are supported under approximate binning. The rank conditions are structurally stable: because the relevant singular values are well above the SVD tolerance ($10^{-8}\times\sigma_1$), small perturbations in $\hat r(A)$ do not change the ranks.

We note that the $k$-means diagnostic rests on an approximation that lies outside the formal identification theory: Theorem~\ref{thm:fixed_x_id} requires menus in a cell to share the \emph{exact} feature value $x$, whereas clustering only ensures approximate homogeneity within each cell. This is a standard finite-sample device analogous to binning in semiparametric estimation, and its validity depends on $\omega(x)$ varying smoothly over the feature support. The overidentification $J$-test reported in Section~\ref{subsec:two_step_results} provides a complementary, estimation-based check on the adequacy of the gate restriction that does not rely on exact-equality cells.

\begin{table}[t]
\centering
\caption{\label{tab:global_id}Global identification verification (D3). Conditions (G1) and (G2) of Corollary~\ref{cor:global_id_from_fixedx} are checked on the Choices13k data using $k$-means clustering on interpretable $z$-features ($d=11$).}
\begin{tabular}{lrl}
\toprule
Condition & Value & Status \\
\midrule
Effective feature dimension $d_{\mathrm{eff}}$ & 11 & \\
(G1) Cells with $\mathrm{rank}(H^{(k)})=|\mathcal F|-1$ & 41 / 12 needed & $\checkmark$ \\
(G2) $\mathrm{rank}([1, z])$ & 12 / 12 needed & $\checkmark$ \\
\midrule
Globally identified &  & $\checkmark$ \\
\bottomrule
\end{tabular}
\end{table}

\subsection{Predictive validation}\label{subsec:results_perf}

Table~\ref{tab:cv_results} reports mean test MSE across 50 random train/test splits. Rule-gating substantially outperforms all structured models (Neural EU/PT/CPT, Value-Based) and the flexible Context-Dependent model. Only MOT achieves marginally lower MSE, with the gap closing to about $2.5\%$ of the MSE level.\footnote{The rule-gating model's effective inputs include both $z(A)$ and the $24$-dimensional rule indicators. The neural benchmarks use the raw $40$-dimensional encoding $\psi(A)$. The feature-based comparison in Appendix Table~\ref{tab:sklearn} provides the closest feature-controlled evaluation. A GBT trained on summary statistics achieves MSE $0.0096$, below both rule-gating and MOT; the interpretability cost is non-trivial but the RRM's decomposition is structural, not post-hoc, which is what permits the mechanism-level analysis that follows.}

Figure~\ref{fig:cv_performance_panel} visualizes the comparison with split-variability bands. Figure~\ref{fig:learning_panel} reports a learning-curve diagnostic: rule-gating outperforms MOT at small training fractions; MOT pulls ahead only above roughly $70\%$ of the data, consistent with its larger number of free parameters. The trial-weighted metric $\mathrm{MSE}_w$ (Table~\ref{tab:cv_results}) preserves all pairwise orderings; Appendix Tables~\ref{tab:paired_comparisons} and~\ref{tab:model_ci} report split-level comparisons.

\begin{table}[h!]
\caption{\label{tab:cv_results}Out-of-sample predictive performance (50 random splits). Learning rates are selected by inner-fold validation MSE (Pass~A); test MSE and $\mathrm{MSE}_w$ are evaluated at the selected learning rate on the held-out test set (Pass~B). $\mathrm{MSE}_w$ denotes the trial-weighted MSE defined in \eqref{eq:mse_w}.}
\centering
\begin{tabular}[t]{lrrrr}
\toprule
Model & Mean MSE & SD & Mean $\mathrm{MSE}_w$ & Best LR\\
\midrule
MOT & 0.01139 & 0.00056 & \multicolumn{1}{c}{n.a.} & 0.001\\
Rule-gating & 0.01168 & 0.00049 & 0.01155 & 0.010\\
Context-dependent & 0.01344 & 0.00081 & 0.01330 & 0.010\\
Value-based & 0.01921 & 0.00081 & 0.01904 & 0.010\\
Neural PT & 0.02064 & 0.00085 & 0.02050 & 0.010\\
Neural CPT & 0.02145 & 0.00093 & 0.02134 & 0.010\\
Neural EU & 0.02215 & 0.00085 & 0.02205 & 0.010\\
\bottomrule
\end{tabular}\\[3pt]
{\footnotesize\textit{Note:} $^\dagger$MOT learning rate is selected by test MSE following the published protocol of \citet{peterson2021} (see Section~\ref{subsec:cv}); the HURD framework does not produce per-menu predictions, so $\mathrm{MSE}_w$ is unavailable for MOT in cross-validation.}
\end{table}

\begin{figure}[t]
    \centering
    \begin{subfigure}[t]{0.49\linewidth}
        \centering
        \includegraphics[width=\linewidth]{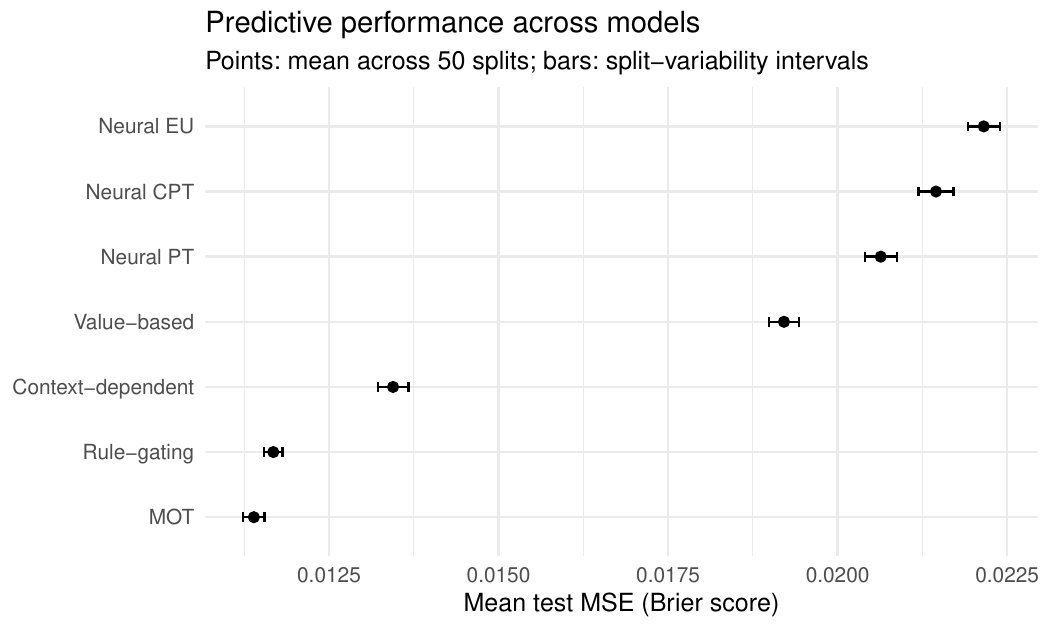}
\caption{Mean test MSE across 50 random splits (split-variability band).}
        \label{fig:cv_performance_panel}
    \end{subfigure}\hfill
    \begin{subfigure}[t]{0.49\linewidth}
        \centering
        \includegraphics[width=\linewidth]{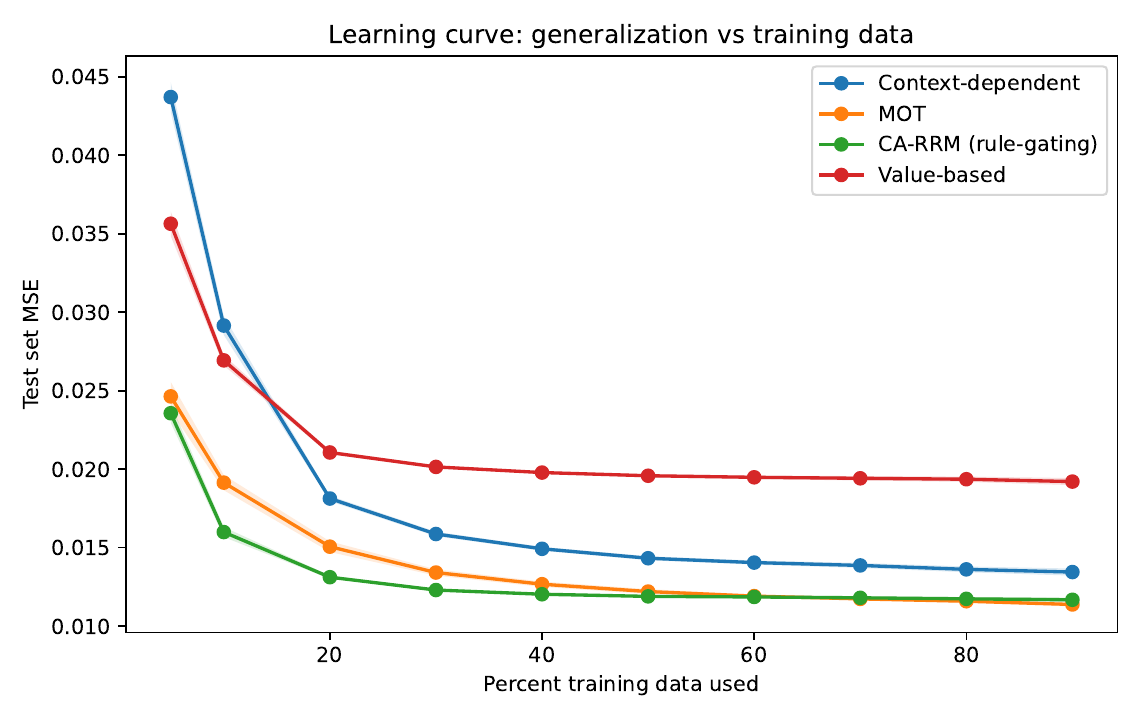}
\caption{Learning curve: test MSE vs.\ training fraction (split-variability band).}
        \label{fig:learning_panel}
    \end{subfigure}
    \caption{\label{fig:perf_two_panels}Predictive performance. Panel (a) compares out-of-sample MSE across models. Panel (b) reports the learning curve, plotting test-set
    MSE as a function of the percent training data used.}
\end{figure}

\subsection{Responsibility weights $w_f$}\label{subsec:responsibility_weights}

Figure~\ref{fig:rule_weights} reports the fitted rule-responsibility weights $w_f$ for the full 12-rule model, computed as average \emph{conditional-on-activity} responsibilities:
\[
w_f=\mathbb E_A\!\left[\tilde q_f(A;\hat\theta)\right],\qquad
\tilde q_f(A;\theta)=\frac{q_f(A;\theta)\,A_f(A)}{\sum_{g\in\mathcal F} q_g(A;\theta)\,A_g(A)}.
\]
Recall that $w_f$ measures the average share of decision mass attributed to rule $f$ \emph{among decisive rules}. We refer to $w_f$ as the \emph{responsibility} (or attribution) weight of rule~$f$, since it reflects the fraction of the predicted choice probability that can be attributed to that rule in the fitted mixture, rather than a direct observation of which rule a decision-maker ``uses.''

Two patterns stand out. First, the attention rules \textsc{A1} and \textsc{A2} absorb the two largest shares of effective responsibility ($w_{\textsc{A1}}\approx 0.280$ and $w_{\textsc{A2}}\approx 0.260$), with a combined attention weight of about $54\%$. This is expected: the attention rules are decisive on essentially all menus and serve as a catch-all channel. Among the other rules, \textsc{SAL} (most salient comparison, $9.4\%$), \textsc{REG} (regret, $7.4\%$), \textsc{MAP} (modal payoff, $6.0\%$), \textsc{MMx} (aspiration, $5.2\%$), and \textsc{MMn} (security, $5.1\%$) carry the largest responsibility weights.
Second, because multiple rules are typically active on the same menu (on average, about $5.8$ rules are both active and consistent with the majority choice direction), these responsibility weights reflect the outcome of competition among substitutable procedures rather than mechanical assignment to the only available rule.

It is useful to separate the library into two conceptual channels. The \emph{structured rules} (\textsc{MMn}, \textsc{MMa}, \textsc{MMx}, \textsc{MAP}, \textsc{SAL}, \textsc{SAL2}, \textsc{REG}, \textsc{REGmed}, \textsc{DIS}, \textsc{DISmed}) encode distinct behavioral mechanisms that generate non-trivial, menu-dependent recommendations. The \emph{attention channel} (\textsc{A1} and \textsc{A2}) captures inattention, hasty responding, or trembling: \textsc{A1} always recommends left and \textsc{A2} always recommends right, regardless of lottery content. These attention rules absorb roughly half of effective responsibility, reflecting the prevalence of default-like behavior in the online (MTurk) environment. Because the marginal left-choice frequency is $\bar p \approx 0.510$ (essentially one-half), this weight cannot be explained by aggregate position bias; rather, the gate allocates attention weight in a menu-dependent way. Part of this weight may, however, reflect absorption of residual prediction error by always-active rules rather than a cleanly identified behavioral channel. Our interpretability claims concern the \emph{structured rules conditional on the inattention channel}. Appendix~\ref{app:attention_removal} confirms that mechanism patterns persist when the attention channel is removed entirely.

The Herfindahl index $\mathrm{HHI}(F)=\sum_f w_f^2\approx 0.173$ yields an effective number of rules $N_{\mathrm{eff}}=1/\mathrm{HHI}\approx 5.8$, reflecting the concentration of responsibility on the attention channel; among the ten structured rules alone, weight is more evenly spread.

\begin{figure}
    \centering
    \includegraphics[width=0.75\linewidth]{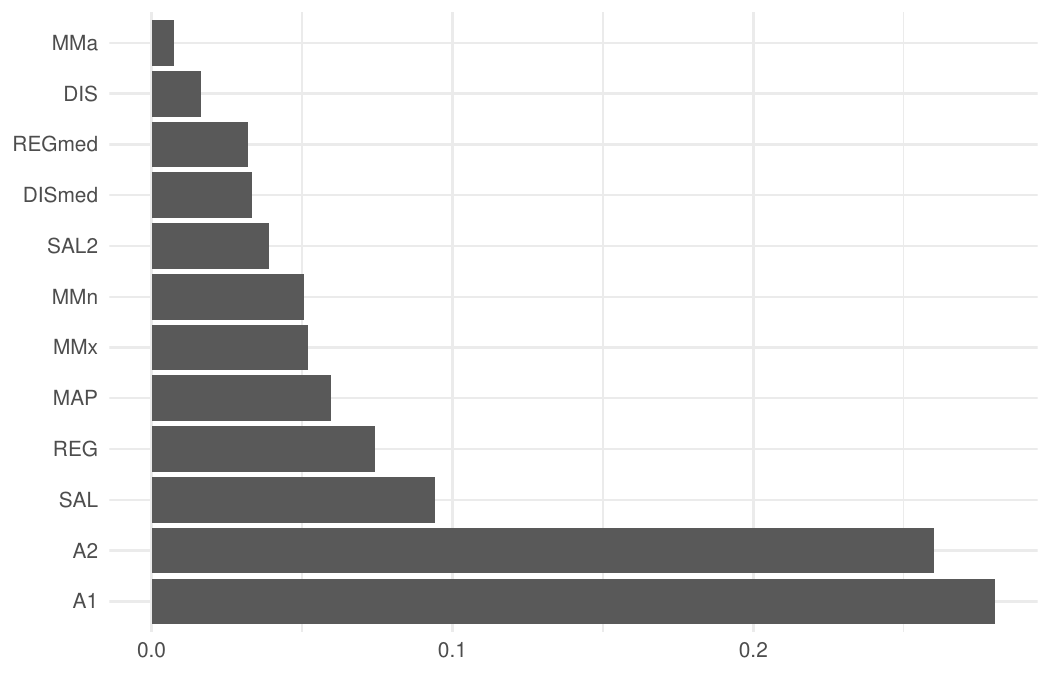}
    \caption{\label{fig:rule_weights}Rule responsibility weights $w_f$ for the full 12-rule library. Each $w_f$ is the average share of predicted choice probability attributed to rule~$f$ among decisive rules.}
\end{figure}

\subsection{Rule concentration}

\begin{figure}[t]
    \centering
   \begin{subfigure}[t]{0.49\linewidth}
        \centering
           \centering
    \includegraphics[width=0.75\linewidth]{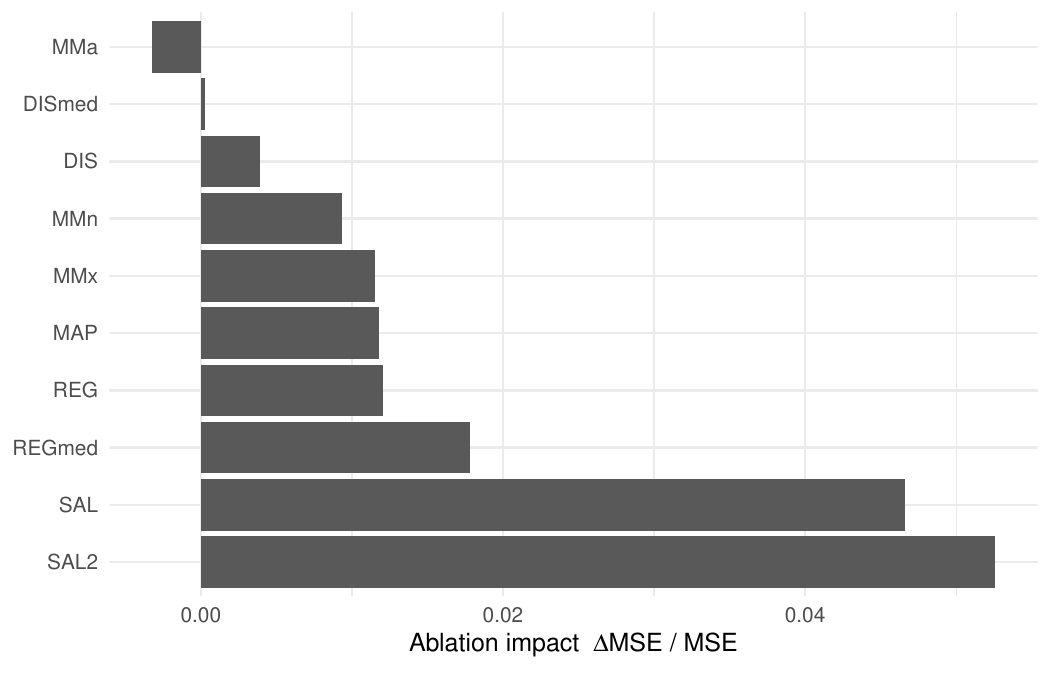}
    \caption{\label{fig:phi} Predictive ablation index $\phi(f)$.}
    \end{subfigure}
    \begin{subfigure}[t]{0.49\linewidth}
        \centering
    \includegraphics[width=0.75\linewidth]{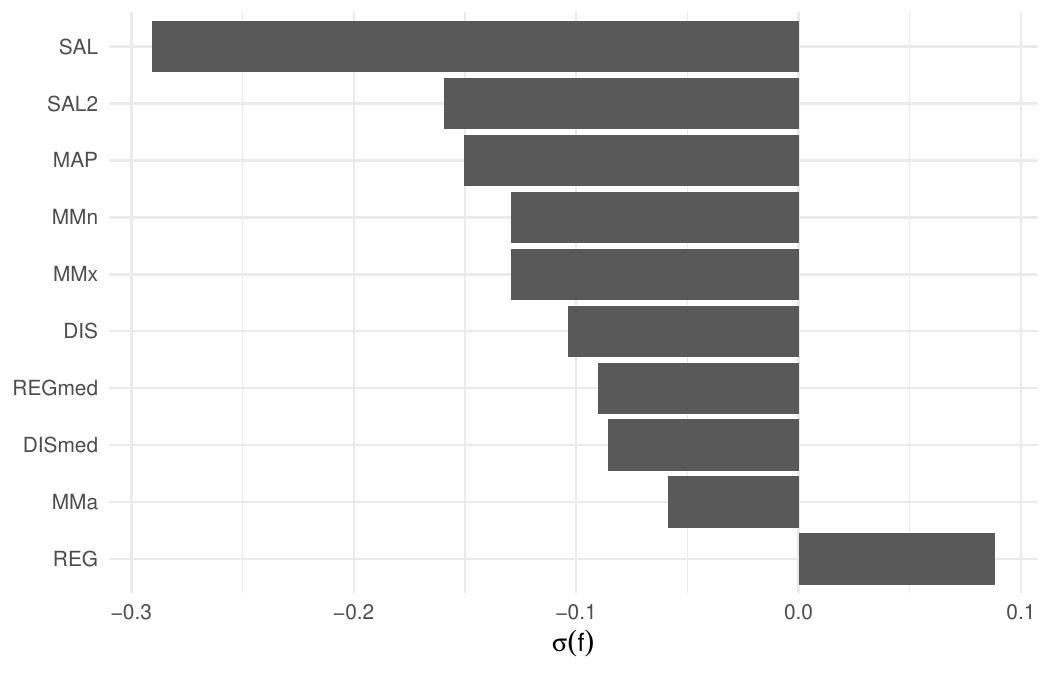}
    \caption{\label{fig:sigma} Concentration Impact $\sigma_N(f)$.}
    \end{subfigure}\hfill
   \caption{Rule diagnostics: ablation of each behavioral rule from the full 12-rule library. Error bars show split-variability intervals ($\pm 1.96 \times \mathrm{SE}$) from fold-to-fold variability across 50 CV splits.}
\end{figure}

We next quantify which procedures are necessary for predictive fit. Figure~\ref{fig:phi} reports the ablation index $\phi(f)$, the relative change in
test MSE when rule $f$ is removed from the full library and the gate is retrained on the reduced set. We ablate each of the 10 behavioral rules (excluding attention rules \textsc{A1}/\textsc{A2}, which serve a structural role ensuring decisiveness at every menu). A small subset of rules accounts for most of the ablation impact:
removing \textsc{SAL2} or \textsc{SAL} produces the largest deterioration, followed by \textsc{REGmed} and \textsc{REG}. \textsc{MAP}, \textsc{MMx}, and \textsc{MMn} have moderate impacts, while \textsc{DIS}/\textsc{DISmed} have small impacts. Notably, \textsc{SAL2} has a relatively high ablation impact despite carrying a lower average responsibility weight than \textsc{SAL}; this illustrates the distinction between $w_f$ (average usage) and $\phi(f)$ (marginal predictive contribution), and suggests that the second-most-salient comparison provides incremental content not replicated by other rules. \textsc{MMa} has a slightly \emph{negative} ablation index, indicating that removing this low-weight rule marginally improves fit - a standard regularization effect when a near-redundant rule is dropped.

Ablations identify rules that matter for accuracy, but they do not indicate whether a rule acts as a \emph{unifier} (concentrating effective responsibility on a small repertoire)
or as a \emph{substitute} that diversifies the fitted mixture.  To capture this, in Figure \ref{fig:sigma}, we report the relative change in the effective number of rules $\sigma_N(f)$. We find that $\sigma_N(f)$ is negative for all non-attention rules except \textsc{REG}, which has a small positive value. Removing \textsc{REG} slightly \emph{increases} the effective number of rules, reflecting its specialized role: \textsc{REG} is active on only about $30\%$ of menus, and its removal frees weight to spread more evenly across the broadly active rules. For all other rules, $\sigma_N(f)<0$: removing a rule concentrates the mixture. The magnitude is largest for the same procedures that have large ablation impacts, indicating that these rules both attract mass and shape the overall mixture structure.

Overall, salience-based and regret-type comparisons account for the largest ablation impacts, as evidenced by their large $\phi(f)$, while the responsibility weights $w_f$ additionally highlight \textsc{SAL} and \textsc{REG} as the most prominent structured rules. Most remaining rules act as substitutes that diversify the fitted mixture, as indicated by their negative $\sigma_N(f)$ values. Appendix Table~\ref{tab:phi_ci} reports split-variability intervals for the absolute ablation effect $\Delta\mathrm{MSE}$, confirming that the main rankings are stable. More broadly, fold-level stability of the decomposition objects is reported in Appendix Table~\ref{tab:decomposition_stability}: the rule rankings by $w_f$ are highly stable across independent training samples (mean Spearman $0.93$); $\phi(f)$ rankings show moderate stability (mean Spearman $0.44$), reflecting the greater sensitivity of ablation effects to sample composition. This decomposition provides a transparent characterization of which bounded-rationality mechanisms are empirically relevant in a large dataset of choice under risk - a level of structural interpretability not available from the more flexible predictors.

On average, about $5.8$ rules are both active and consistent with the majority choice direction per menu, and about $9.5$ rules are active; menus pinning down a unique decisive rule are rare ($0.16\%$). The gate therefore learns selection among genuinely competing procedures (see Appendix Figure~\ref{fig:overlap_hist} for the full distribution).

\subsection{Comparative statics in rule responsibility}

The figures in this section plot effective responsibility weights $\tilde q_f(A;\theta)$, which condition on rule activity (Equation~\eqref{eq:qtilde_def}). Changes in $\tilde q_f$ across menus reflect two channels: (i)~genuine shifts in the latent gate propensity $q_f(A;\theta)$ and (ii)~mechanical denominator effects from the composition of the active set $\mathcal F_A$, which varies across menus. We have verified that the qualitative patterns are preserved when using the latent gate outputs $q_f(A;\theta)$ directly (before conditioning on activity), confirming that the reported patterns reflect genuine variation in the learned gate rather than mechanical composition effects.

\subsubsection*{Complexity}

A central motivation for rule-gating is that a fixed decision rule is too coarse to account for behavior across the wide range of menus in
\texttt{choices13k}. The substantive hypothesis is therefore not merely that people are noisy, but that they adaptively select different procedures across menus. To connect this hypothesis to the emerging literature on complexity and procedure choice, we rely on an objective notion of menu complexity that is defined directly from the two lotteries in the menu and is available ex ante. This notion is designed to capture when a choice problem requires trading off state-by-state advantages and disadvantages across the two options-precisely the situations where different heuristics can plausibly disagree and where a gate has
something non-trivial to learn.

In binary lottery choice, a particularly natural notion is \emph{tradeoff complexity}. Formally, for a menu $A=\{L^1,L^2\}$, let $F_{1}$ and $F_{2}$ denote the
cumulative distribution functions of the two objective lotteries (after placing their supports on a common ordered grid). Define the CDF distance
\[
\Delta_{\mathrm{CDF}}(A)\;:=\;\int_{\mathbb R}\big|F_{1}(x)-F_{2}(x)\big|\,dx,
\]
and the associated \emph{excess dissimilarity}
\[
d(A)\;:=\;\Delta_{\mathrm{CDF}}(A)-\big|EV(L^1)-EV(L^2)\big|.
\]
The subtraction removes the mechanical effect of a pure value gap and isolates the extent to which the two lotteries ``cross'' in payoff space, i.e., the extent to
which the comparison requires genuine tradeoffs rather than a near-unanimous ordering. In the empirical analysis we work with the log-transformation
$\mathrm{TC}(A):=\log(1+d(A))$ and plot mean effective responsibility weights across complexity deciles.

\begin{figure}
    \centering
    \includegraphics[width=0.65\linewidth]{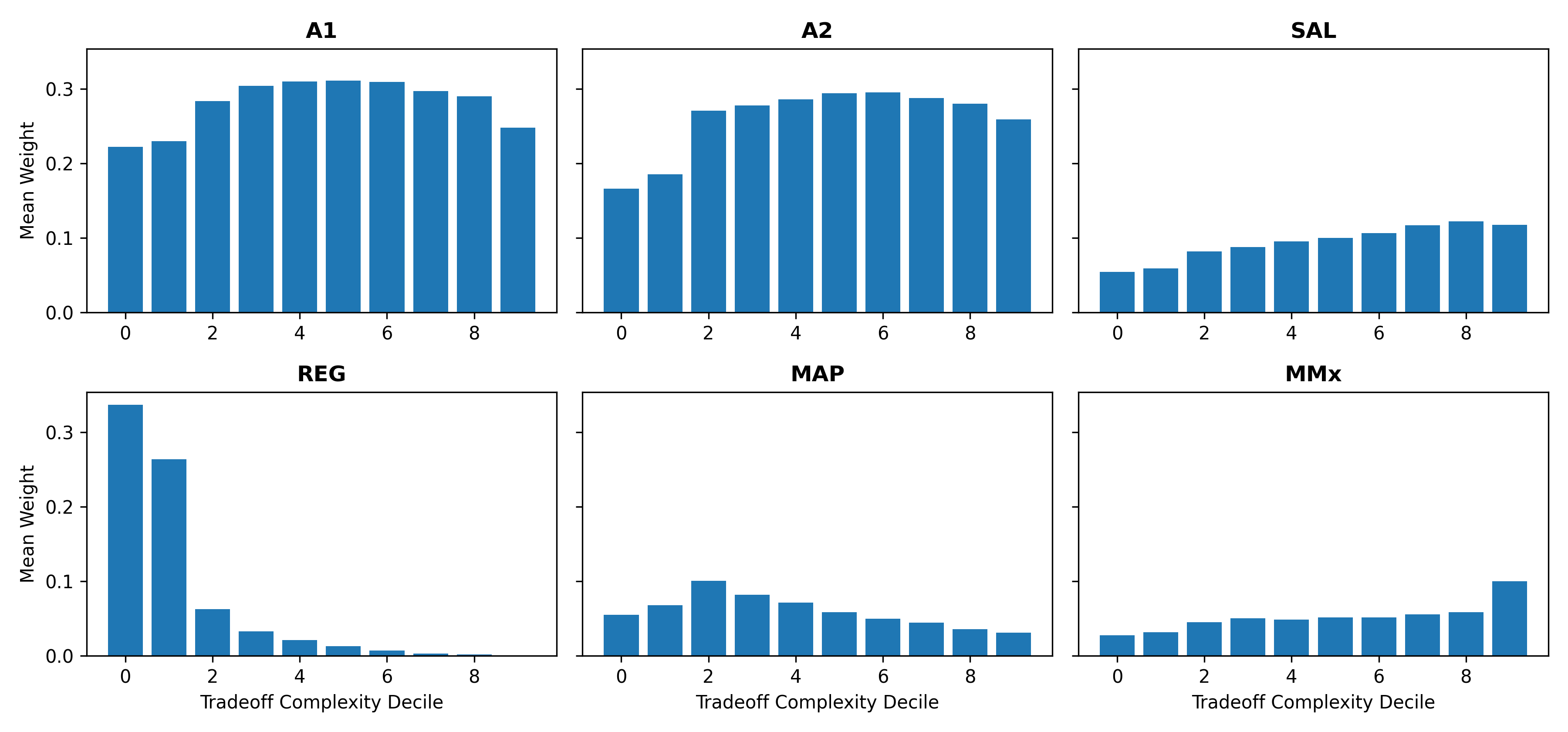}
    \caption{\label{fig:context_complexity_top6}Mean effective responsibility weights across deciles of tradeoff complexity $\mathrm{TC}(A)$ for the highest-weight rules.}
\end{figure}

Figure~\ref{fig:context_complexity_top6} shows clear and interpretable shifts in rule selection along tradeoff complexity. As $\mathrm{TC}(A)$ increases, the two attention rules (\textsc{A1}/\textsc{A2}) display a pronounced increase in mean effective weight, consistent with the view that more complex menus are associated with a
greater reliance on coarse inattention in the learned CA-RRM. Among the structured procedures, the most striking pattern is \textsc{REG} (regret): it carries its highest weight in the lowest-complexity deciles and drops steeply to near zero at high complexity, indicating that regret-type evaluation is concentrated on simple menus where state-by-state comparison is tractable. \textsc{SAL} (salience) rises with complexity, while \textsc{MAP} (modal payoff) displays an inverted-U pattern, peaking at moderate complexity and declining at both extremes. \textsc{MMx} (aspiration) shows a mild increase across deciles. Overall, increasing tradeoff complexity tilts the fitted mixture away from evaluative comparisons (regret, modal outcome) and toward attention-based procedures and salience.

\subsubsection*{Risk} 

A second natural driver is the risk profile of the menu. Different procedures emphasize different aspects of risk, such as dispersion,
downside exposure, and sensitivity to rare events. As a first-pass risk descriptor available directly from the menu, we consider the absolute asymmetry in dispersion across the
two lotteries,
\[
\mathrm{RiskAsym}(A)\;:=\;|\sigma(L^1)-\sigma(L^2)|,
\]
where $\sigma(L)$ denotes the standard deviation of lottery $L$. This measures how different the two lotteries are in terms of spread, regardless of which lottery is riskier. Figure~\ref{fig:context_risk_top6} plots mean responsibility weights across deciles of this measure.

\begin{figure}[t]
    \centering
    \includegraphics[width=0.65\linewidth]{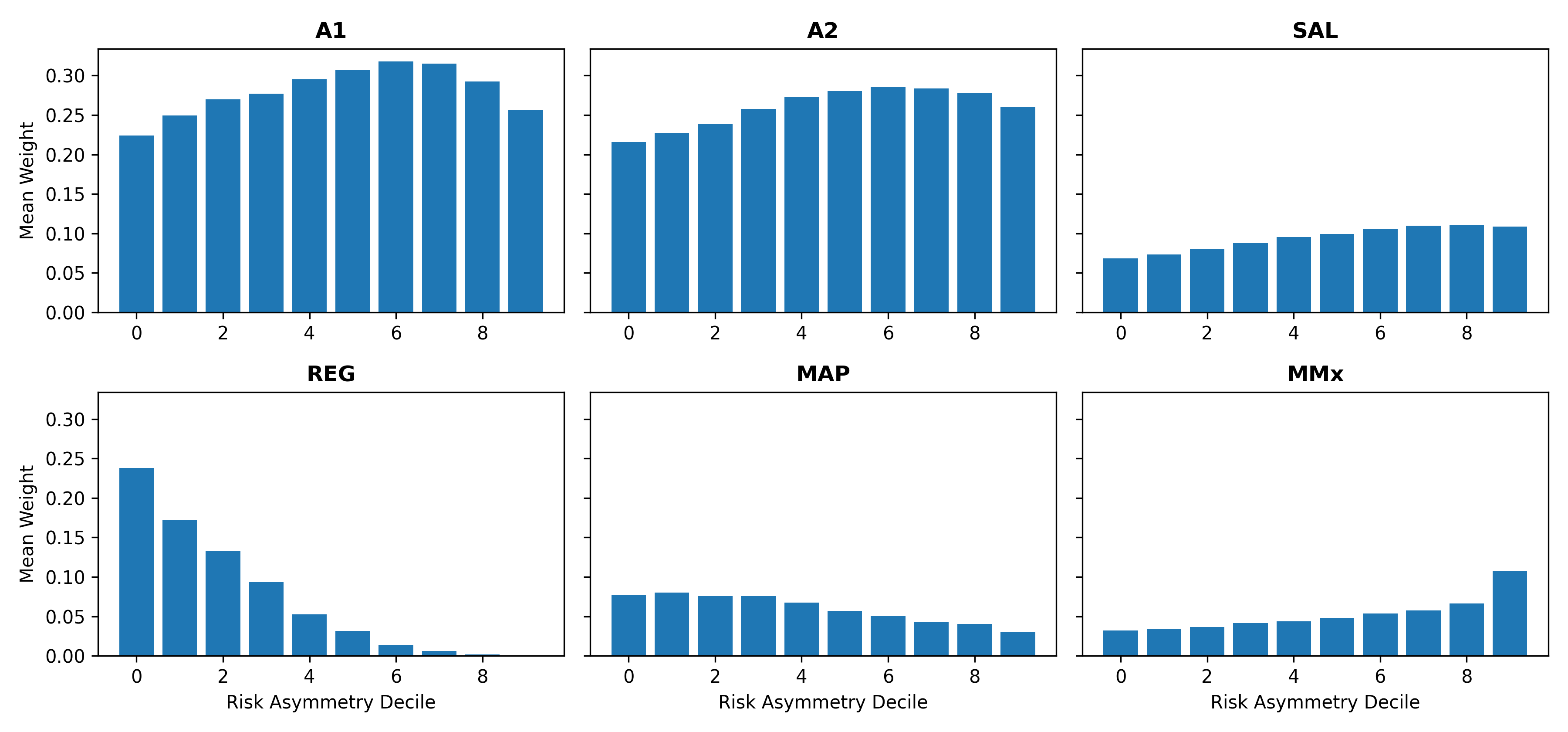}
    \caption{\label{fig:context_risk_top6}Mean effective responsibility weights across deciles of risk asymmetry $\mathrm{RiskAsym}(A)$ for the highest-weight rules.}
\end{figure}

Figure~\ref{fig:context_risk_top6} provides a complementary view based on dispersion asymmetry. The most striking pattern involves \textsc{REG} (regret): its weight is highest when the two lotteries have similar dispersion (low asymmetry) and declines steeply as the gap in spread widens. By contrast, the attention defaults \textsc{A1} and \textsc{A2} increase monotonically with risk asymmetry, suggesting that large differences in spread push the gate toward inattention. \textsc{SAL} and \textsc{MMx} show mild increases, while \textsc{MAP} declines gently. Thus, in menus where the two options differ markedly in spread, the gate shifts mass away from regret-type evaluation and toward attention defaults.

These diagnostics show that rule selection varies systematically with salient menu primitives: along complexity, the mixture tilts toward inattention and salience-based comparison and away from regret-type and modal-outcome evaluation; along dispersion asymmetry, regret-type comparison declines steeply while attention defaults rise. Menu-dependent selection across a small repertoire of procedures is economically meaningful.

\subsection{Completeness and Restrictiveness}

We supplement the absolute MSE comparisons with two diagnostics: \emph{completeness} (how much of a flexible benchmark's variation is captured by a structured model) and \emph{restrictiveness} (how tightly a model class constrains the mappings it can fit). Together, they help separate performance gains reflecting substantive structure from gains due to sheer flexibility.

\subsubsection*{ML Completeness Score}

Following \citet{peysakhovich2017}, \citet{fudenberg2022}, and \citet{fudenberg2022_jpe}, we define the ML completeness of model $g$ relative to Neural EU (baseline) and MOT (flexible benchmark):
\begin{equation}\label{eq:ml_completeness}
\mathrm{Comp}(g)
:=\frac{\mathrm{MSE}(g_{\mathrm{EU}})-\mathrm{MSE}(g)}{\mathrm{MSE}(g_{\mathrm{EU}})-\mathrm{MSE}(g_{\mathrm{MOT}})}.
\end{equation}
By construction, $\mathrm{Comp}(g_{\mathrm{EU}})=0$ and $\mathrm{Comp}(g_{\mathrm{MOT}})=1$; values in $(0,1)$ indicate partial recovery. Rule-gating recovers approximately $97\%$ of the baseline-to-MOT gap, compared to $\approx 27\%$ for the Value-based architecture.

\subsubsection*{Restrictiveness}

A complementary question is \emph{restrictiveness}: how tightly a model class constrains the set of menu-to-choice mappings it can fit \citep{selten91,fudenberg26}. Following \citet{fudenberg26}, we define permuted targets $\hat p^{\pi}(A):=\hat p(\pi(A))$ for a random permutation $\pi$ of the training menus, and measure permutation-fit restrictiveness as
\begin{equation}\label{eq:restrictiveness}
\mathrm{Restr}(g)
:=\mathbb E_{\mathcal T,\pi}\!\left[
\frac{\mathrm{MSE}^{\pi}_{\mathrm{train}}(g)}{\mathrm{MSE}^{\pi}_{\mathrm{train}}(\mathrm{const})}
\right],
\end{equation}
averaged over $K=50$ splits and $P=10$ permutations per split. Values near~$1$ indicate the model cannot fit permuted targets better than a constant predictor (highly restrictive); values well below~$1$ indicate capacity to absorb structureless patterns.

Figure~\ref{fig:frontier} plots models in the (completeness, restrictiveness) plane. Rule-gating occupies the upper-right region: it recovers nearly all of MOT's predictive advantage while remaining essentially fully restrictive ($\mathrm{Restr}\approx 1.0$). MOT is markedly less restrictive ($\mathrm{Restr}\approx 0.2$), indicating substantial capacity to absorb structureless variation.

\begin{figure}[t]
\centering
\includegraphics[width=1\linewidth]{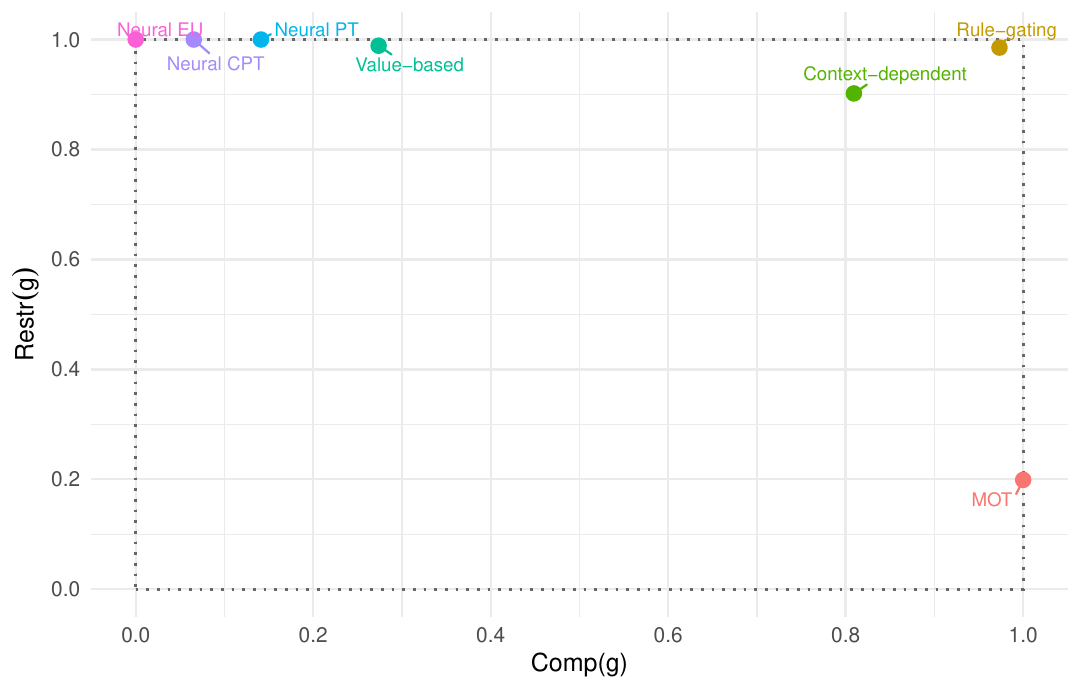}
\caption{\label{fig:frontier}Completeness--restrictiveness plot. The horizontal axis is ML completeness (Equation~\eqref{eq:ml_completeness});
the vertical axis is permutation-fit restrictiveness (Equation~\eqref{eq:restrictiveness}).}
\end{figure}

\subsection{Portability}\label{subsec:portability}

Finally, a central motivation for interpretable decision procedures is that they are intended to capture stable
regularities of behavior, not merely to interpolate within a single dataset. A natural final check is therefore
\emph{portability}: does a predictor trained on one large-scale choice environment retain predictive content when
exported - without re-tuning - to an independent dataset with different subjects and a different menu sample?

To assess portability, we evaluate models trained on \texttt{choices13k} \citep{peterson2021} on the independent CPC18 dataset \citep{erev_ert_plonsky_2017,plonsky2017_cpc18}.\footnote{The data can be downloaded at \url{https://cpc-18.com/data/}.}
We restrict to risk problems, yielding $686$ subjects, $|\Acal^{\mathrm{CPC}}|=180$ distinct menus, and $354{,}400$ total trials.\footnote{In addition to menu-level aggregation, we also report a trial-weighted evaluation.}

For each model, we freeze \texttt{choices13k}-estimated parameters and evaluate on CPC18 menus using
\begin{equation}\label{eq:mse_port}
\mathrm{MSE}_{\mathrm{CPC}}(g_{\rightarrow})
:=\frac{1}{|\Acal^{\mathrm{CPC}}|}\sum_{A\in\Acal^{\mathrm{CPC}}}\Big(g_{\rightarrow}(A)-\hat p_{\mathrm{CPC}}(A)\Big)^2.
\end{equation}
The outcome rescaling factor is computed from \texttt{choices13k} and applied identically to CPC18.\footnote{Both datasets use a common monetary unit, so the rescaling simply puts outcomes on a common numerical scale.} No hyperparameters are tuned on CPC18. Because CPC18 records individual-level trial data, we also evaluate at the trial level using the Brier score and Bernoulli log-loss.

\paragraph{Portability results.}
Table~\ref{tab:portability} reports results. MOT---the best in-sample performer---deteriorates sharply under transfer, suggesting that much of its advantage is dataset-specific. Rule-gating retains strong predictive content: CPC18 MSE $0.0148$, ahead of the context-dependent model ($0.0173$) and far ahead of the structured neural benchmarks ($\approx 0.032$). Rule-gating therefore offers the best portability among all models \emph{with} an interpretable behavioral decomposition.

\begin{table}[t]
\centering
\caption{Portability to CPC18: models trained on \texttt{choices13k} and evaluated on CPC18 without re-tuning. Trial-level Brier score and log-loss use individual choice observations ($n=354,400$); menu-level MSE aggregates to the menu level ($n=180$).}
\label{tab:portability}
\begin{tabular}{lrrr}
\toprule
Model & $\mathrm{Brier}_{\mathrm{trial}}$ & $\mathrm{LogLoss}_{\mathrm{trial}}$ & $\mathrm{MSE}_{\mathrm{menu}}$ \\
\midrule
Rule-gating & 0.1952 & 0.5705 & 0.0148 \\
Context-dependent & 0.1977 & 0.5763 & 0.0173 \\
Neural CPT & 0.2111 & 0.6111 & 0.0319 \\
Neural EU & 0.2111 & 0.6115 & 0.0318 \\
Neural PT & 0.2120 & 0.6135 & 0.0325 \\
Value-based & 0.2122 & 0.6146 & 0.0323 \\
MOT & 0.2713 & 0.7513 & 0.0943 \\
\bottomrule
\end{tabular}\\[3pt]
{\footnotesize\textit{Note:} $\mathrm{Brier}_{\mathrm{trial}}$ and $\mathrm{LogLoss}_{\mathrm{trial}}$ evaluate each model's predicted menu-level probability against individual binary choices. All models are trained once on the full Choices13k sample and evaluated on CPC18 without re-tuning.}
\end{table}

\subsection{Two-step econometric estimates}\label{subsec:two_step_results}

Table~\ref{tab:two_step} reports the two-step econometric estimates. The first stage groups menus into $47$ cells by $k$-means on $z(A)$ and recovers cellwise normalized weights via constrained least squares on the $H^{(k)}$ matrices; the second stage regresses log-weights on cell centroids to recover the affine gate coefficients (see Appendix~\ref{app:two_step_details} for details). The results should be read as an econometric cross-check motivated by the identification theory.

The two-step responsibility weights show moderate-to-good agreement with the MSE-fitted decomposition (Spearman $0.70$, Pearson $0.40$). The most notable discrepancy involves the attention rules: \textsc{A1} carries the largest MSE weight ($w_f=0.28$) but drops to fifth under the two-step procedure ($w_f=0.11$), while \textsc{SAL} is the most prominent two-step rule ($w_f=0.23$) but ranks third under MSE. This is expected: the attention rules are one-sided and contribute no within-cell variation in decisive-side patterns, so their two-step weights are identified only residually.

The overidentification $J$-test rejects the affine gate restriction at the $1\%$ level for most rules, indicating that $\log\omega_f(x)=\alpha_f+\beta_f^\top z$ is too parsimonious to capture the full pattern of cellwise weight variation. This is not unexpected given the test's 35 degrees of freedom. Importantly, Theorem~\ref{thm:fixed_x_id} is fully nonparametric: the cellwise identification of relative weights holds for any functional form. The $J$-test rejection suggests that a richer second-step specification could reduce or eliminate rejection while preserving the first-step identification logic.

\begin{table}[t]
\centering
\caption{\label{tab:two_step}Two-step econometric estimator vs.\ MSE-fitted responsibility weights $w_f$. The two-step estimator recovers cellwise rule weights by constrained least squares on $H^{(k)}$, then maps log-weights into affine gate parameters. Bootstrap standard errors (100 resamples) in parentheses. The $J$-statistic tests the affine gate restriction ($\chi^2_{35}$).}
\begin{tabular}{lcccc}
\toprule
Rule & $w_f$ (two-step) & $w_f$ (MSE) & Difference & $J$-stat \\
\midrule
\textsc{MMn} & 0.131 (0.062) & 0.051 & $+0.080$ & 293.2$^{***}$ \\
\textsc{MMa} & 0.000 (0.012) & 0.008 & $-0.008$ & 45.4 \\
\textsc{MMx} & 0.170 (0.063) & 0.052 & $+0.118$ & 252.0$^{***}$ \\
\textsc{MAP} & 0.023 (0.062) & 0.060 & $-0.037$ & 357.6$^{***}$ \\
\textsc{SAL} & 0.228 (0.070) & 0.094 & $+0.134$ & 197.0$^{***}$ \\
\textsc{SAL2} & 0.072 (0.034) & 0.039 & $+0.033$ & 210.1$^{***}$ \\
\textsc{REG} & 0.099 (0.030) & 0.074 & $+0.025$ & 152.2$^{***}$ \\
\textsc{REGmed} & 0.039 (0.051) & 0.032 & $+0.007$ & 212.3$^{***}$ \\
\textsc{DIS} & 0.004 (0.022) & 0.017 & $-0.013$ & 186.5$^{***}$ \\
\textsc{DISmed} & 0.001 (0.028) & 0.033 & $-0.032$ & 136.4$^{***}$ \\
\textsc{A1} & 0.109 (0.070) & 0.280 & $-0.171$ & 26.1 \\
\textsc{A2} & 0.123 (0.043) & 0.260 & $-0.137$ & --- \\
\midrule
\multicolumn{5}{l}{Pearson correlation: 0.401; Spearman rank correlation: 0.699} \\
\multicolumn{5}{l}{$^{***}p<0.01$; $H_0$: $\log\omega_f(x)$ is affine in $x$} \\
\bottomrule
\end{tabular}
\end{table}

\section{Further Robustness}\label{sec:robustness}

We subject the framework to a battery of robustness checks, reported in full in the Online Appendix. Cross-fitted library selection chooses rule subsets only on training folds: restricting to $k=7$ rules retains the bulk of full-library accuracy, and the selected rules match those receiving the highest responsibility weights (Appendix Table~\ref{tab:crossfitted}). Removing the two attention rules raises MSE to $0.0201$, still outperforming Neural EU/PT/CPT, confirming that the structured rules alone carry substantial predictive content; the comparative-statics patterns persist qualitatively (Appendix~\ref{app:attention_removal}). Varying the FSD activity discipline shows that over-restricting the activity set is more costly than removing it, so the dominance-based activation is best understood as an interpretive discipline rather than a predictive one (Appendix Table~\ref{tab:epsilon_fsd}). A placebo test scrambling rule indicators across menus produces substantially higher MSE, confirming that the indicators carry genuine, menu-specific content (Appendix Table~\ref{tab:placebo}). Feature-based baselines (fractional logit, GBT) on rule indicators underperform the gated architecture, while GBT on raw summary statistics outperforms all structured models but provides no behavioral decomposition (Appendix Table~\ref{tab:sklearn}). Finally, a raw-encoding gate using $\psi(A)\in\mathbb{R}^{40}$ instead of $z(A)$ achieves MSE $0.0139$ --- about $19\%$ worse --- confirming that the interpretable summary statistics capture decision-relevant variation without overfitting (Appendix Figure~\ref{fig:gate_comparison}).

\section{Conclusion}\label{sec:conclusion}

This paper studies stochastic choice under latent procedural heterogeneity. We develop a Random Rule Model (RRM) in which observed choice probabilities arise from switching among a small library of transparent deterministic decision procedures, with environment-dependent rule weights. The main contribution is an identification-and-estimation framework for recovering these latent procedural weights from menu-level choice frequencies. Identification has a two-step structure: within-feature variation in decisive-side patterns identifies relative rule weights, while cross-feature affine richness identifies the gate parameters. A constructive two-step estimator then supports inference for economically meaningful rule-level functionals.

Applied to risky choice, the framework uncovers substantial and systematic procedural heterogeneity. Salience-based comparisons, regret-type rules, and modal-outcome comparisons receive the largest responsibility weights among structured procedures, and rule responsibility shifts predictably with menu characteristics such as tradeoff complexity and dispersion asymmetry. These comparative statics provide an empirical characterization of risky choice in terms of changing procedure use rather than only reduced-form choice frequencies.

The predictive exercises support the validity of the representation. Rule-gating substantially outperforms structured neural benchmarks, remains highly restrictive under permutation-fit diagnostics, and retains predictive content on an independent dataset. At the same time, the evidence from varying the activity discipline shows that the activity filter is best understood as an attribution operator that restricts behavioral responsibility to rules that are actually decisive, rather than as the primary source of fit.

On the econometric side, the two-step overidentification test rejects the affine gate specification for most rules, indicating that the parsimonious linear-index parameterization does not fully capture how rule weights vary across menu characteristics. Since the first-stage identification is nonparametric, a natural extension is to replace the affine second step with a richer specification. More broadly, the paper suggests that stochastic choice data can be used to recover not only how often each option is chosen, but also how populations reallocate across decision procedures as environments change. This makes procedural heterogeneity a tractable econometric object, and opens the door to applications in consumer choice, dynamic choice, and strategic environments where behavior may be governed by shifting mixtures of transparent rules.

\bibliographystyle{aea}
\bibliography{bibliography}

\bigskip \bigskip
\appendix
\onehalfspace
\newpage

\clearpage

\begin{LARGE}
    \noindent\textbf{Online Appendix}
\end{LARGE}

\setcounter{page}{1}
\setcounter{table}{0}
\renewcommand{\thetable}{A.\arabic{table}}
\setcounter{figure}{0}
\renewcommand{\thefigure}{A.\arabic{figure}}
\setcounter{section}{0}
\setcounter{equation}{0}
\renewcommand*{\theequation}{A.\arabic{equation}}
\setcounter{theorem}{0}
\renewcommand{\thetheorem}{A.\arabic{theorem}}
\setcounter{assumption}{0}
\renewcommand{\theassumption}{A.\arabic{assumption}}
\setcounter{remark}{0}
\renewcommand{\theremark}{A.\arabic{remark}}

\section{General Procedural Mixture Framework}\label{app:general_framework}

This section presents the Conditional Activity Random Rule Model (CA-RRM) and its identification theory in a general discrete choice environment with finitely many alternatives per menu. The binary lottery application in the main text is a special case obtained by setting $|A|=2$ for all menus and instantiating the dominance preorder as first-order stochastic dominance.

\subsection{Choice environment}

Let $\mathcal{M}$ denote a finite collection of \emph{menus} (choice sets). Each menu $A\in\mathcal{M}$ is a finite set of alternatives with $|A|=n_A\ge 2$.
For each menu $A$ and alternative $i\in A$, let $p(i\mid A)\in(0,1)$ denote the population choice probability, with $\sum_{i\in A}p(i\mid A)=1$.

\begin{assumption}[Interior probabilities]\label{ass:interior}
For every $A\in\mathcal{M}$ and $i\in A$, $p(i\mid A)>0$.
\end{assumption}

This ensures that log-ratios of choice probabilities are well defined.

\subsection{Procedural library and activity}

Let $\mathcal{F}=\{f_1,\dots,f_F\}$ be a finite library of \emph{deterministic decision procedures} (rules). Each rule $f\in\mathcal{F}$ is equipped with:
\begin{enumerate}[label=(\roman*),leftmargin=*,itemsep=0.2em]
\item An \emph{activity set} $\mathcal{A}_f\subseteq\mathcal{M}$: the menus at which rule $f$ delivers a strict ranking of alternatives. Define the activity indicator $\Ind_f(A):=\Ind\{A\in\mathcal{A}_f\}$.
\item A \emph{recommendation mapping} $r_f:\mathcal{A}_f\to\bigcup_{A\in\mathcal{M}}A$ such that $r_f(A)\in A$ for each $A\in\mathcal{A}_f$: the unique alternative that rule $f$ recommends at menu $A$ when active.
\end{enumerate}

Activity is determined by an exogenous criterion (e.g., stochastic dominance of perceived lotteries in risky choice, present-value dominance in intertemporal choice, or welfare dominance in allocation problems). The key discipline is that a rule contributes to predicted choice probabilities \emph{only} at menus where it is active.

For each menu $A$ and alternative $i\in A$, define the \emph{recommending set}
\begin{equation}\label{eq:rec_set}
\mathcal{I}_i(A):=\{f\in\mathcal{F}:\Ind_f(A)=1,\;r_f(A)=i\},
\end{equation}
the set of active rules that recommend alternative $i$ at menu $A$. The set of all active rules at $A$ is $\mathcal{F}_A:=\bigcup_{i\in A}\mathcal{I}_i(A)$.

\subsection{The CA-RRM}

A \emph{Conditional Activity Random Rule Model} assigns to each menu $A$ a probability distribution $\{q_f(A)\}_{f\in\mathcal{F}}\in\Delta^{F-1}$ over rules.
Define the conditional-on-activity weight of rule $f$:
\begin{equation}\label{eq:qtilde_gen}
\tilde q_f(A):=\frac{q_f(A)\,\Ind_f(A)}{\sum_{g\in\mathcal{F}}q_g(A)\,\Ind_g(A)},
\end{equation}
whenever $\mathcal{F}_A\neq\emptyset$. The CA-RRM choice probability is
\begin{equation}\label{eq:carrm_gen}
P(i\mid A;\,q)\;=\;\sum_{f\in\mathcal{I}_i(A)}\tilde q_f(A)
\;=\;\frac{\sum_{f\in\mathcal{I}_i(A)}q_f(A)}{\sum_{g\in\mathcal{F}_A}q_g(A)}.
\end{equation}
Equation~\eqref{eq:carrm_gen} expresses predicted choice as a mixture over deterministic rule recommendations, conditional on activity.

\subsection{Softmax gating parameterization}

The CA-RRM becomes a parametric model when $q(A)$ is generated by a differentiable gate. Let $\psi(A)\in\mathbb{R}^d$ be an observable feature vector for menu $A$. The \emph{softmax gate} sets
\begin{equation}\label{eq:gate_gen}
q_f(A;\theta):=\frac{\exp(u_f(A;\theta))}{\sum_{g\in\mathcal{F}}\exp(u_g(A;\theta))},
\end{equation}
where $u_f(A;\theta):=\alpha_f+\beta_f^\top\psi(A)$ is an affine score, and $\theta:=\{(\alpha_f,\beta_f)\}_{f\in\mathcal{F}}$ collects all gate parameters.

Substituting \eqref{eq:gate_gen} into \eqref{eq:carrm_gen} and canceling the common denominator of the softmax, the CA-RRM choice probability simplifies to
\begin{equation}\label{eq:P_softmax_gen}
P(i\mid A;\theta)\;=\;\frac{\sum_{f\in\mathcal{I}_i(A)}\exp(u_f(A;\theta))}{\sum_{g\in\mathcal{F}_A}\exp(u_g(A;\theta))}.
\end{equation}
This is a \emph{grouped softmax} over active rules, with groups defined by the alternative each rule recommends.

\subsection{Normalization}

The softmax is invariant to adding a common affine function to all scores: if $u_f\mapsto u_f + a + c^\top\psi$ for all $f$, the probabilities in \eqref{eq:P_softmax_gen} are unchanged.
We therefore impose a location normalization by fixing one baseline rule $f_0\in\mathcal{F}$ and setting
\begin{equation}\label{eq:norm_gen}
(\alpha_{f_0},\beta_{f_0})=(0,\mathbf{0}).
\end{equation}
Let $\vartheta\in\mathbb{R}^K$ denote the vector of free parameters, with $K=(F-1)(1+d)$.

\subsection{Identification map and sufficient-variation condition}

\paragraph{Multi-sided menus.}
Say that a menu $A$ is \emph{multi-sided} if at least two distinct alternatives are recommended by some active rule:
\[
\big|\{i\in A:\mathcal{I}_i(A)\neq\emptyset\}\big|\;\ge\;2.
\]
Let $\mathcal{M}_{\mathrm{ms}}\subseteq\mathcal{M}$ denote the set of multi-sided menus (the analogue of two-sided menus in the binary case). On menus that are not multi-sided, all active rules agree, so the CA-RRM prediction is $1$ for the unanimously recommended alternative regardless of $\theta$; such menus carry no identifying information.

\paragraph{Log-ratio representation.}
For each multi-sided menu $A$, restrict attention to alternatives $i\in A$ with $\mathcal{I}_i(A)\neq\emptyset$ (i.e., alternatives recommended by at least one active rule).\footnote{Under correct specification, Assumption~\ref{ass:interior} implies $P(i\mid A)>0$ for all $i\in A$, which requires $\mathcal{I}_i(A)\neq\emptyset$ since \eqref{eq:P_softmax_gen} assigns zero probability to uncovered alternatives. Hence the restriction to covered alternatives is implicit in the joint assumption of correct specification and interior probabilities.} Fix a reference alternative $i_0(A)$ among these. By Assumption~\ref{ass:interior} and \eqref{eq:P_softmax_gen}, the log-ratios
\begin{equation}\label{eq:logratio_gen}
\delta_i(A;\vartheta)\;:=\;\log\frac{P(i\mid A;\vartheta)}{P(i_0(A)\mid A;\vartheta)}
\;=\;\log\!\sum_{f\in\mathcal{I}_i(A)}\exp(u_f(A;\vartheta))\;-\;\log\!\sum_{f\in\mathcal{I}_{i_0}(A)}\exp(u_f(A;\vartheta))
\end{equation}
are well defined for each covered $i\in A\setminus\{i_0(A)\}$.

\paragraph{Identification map.}
Stack the log-ratios across all multi-sided menus and non-reference alternatives into a single vector:
\begin{equation}\label{eq:H_gen}
H(\vartheta)\;:=\;\Big(\delta_i(A;\vartheta)\Big)_{A\in\mathcal{M}_{\mathrm{ms}},\;i\in A\setminus\{i_0(A)\}}\;\in\;\mathbb{R}^D,
\end{equation}
where $D=\sum_{A\in\mathcal{M}_{\mathrm{ms}}}(n_A-1)$. Observing the population choice probabilities $\{p(i\mid A)\}$ on multi-sided menus is equivalent to observing $H(\vartheta)$, because the log-ratio transformation is a diffeomorphism from the interior of the simplex to $\mathbb{R}^{n_A-1}$.

\subsection{Global identification theorem}

\paragraph{Decisive-side indicators and linear restriction.}
For each rule $f$ and each alternative $i\in A$, define the decisive-side indicator
\[
\kappa_f^i(A):=\Ind\{f\in\mathcal{I}_i(A)\},
\]
where $\mathcal{I}_i(A)$ is the recommending set from \eqref{eq:rec_set}.
For each multi-sided menu $A$, fix a reference alternative $i_0(A)$ with $\mathcal{I}_{i_0}(A)\neq\emptyset$ and write the odds ratio
\[
r_i(A):=\frac{P(i\mid A)}{P(i_0(A)\mid A)},\qquad i\in A\setminus\{i_0(A)\}.
\]
Define the positive rule weights $\omega_f(x):=\exp(\alpha_f+\beta_f^\top x)$ for $x=\psi(A)$.
From \eqref{eq:P_softmax_gen}, the model-implied odds ratio satisfies
\[
r_i(A)=\frac{\sum_{f\in\mathcal{I}_i(A)}\omega_f(x)}{\sum_{f\in\mathcal{I}_{i_0}(A)}\omega_f(x)}.
\]
Cross-multiplying and collecting terms yields, for each non-reference alternative $i\in A\setminus\{i_0(A)\}$,
\begin{equation}\label{eq:linear_restriction_gen}
\sum_{f\in\mathcal{F}}\Big(\kappa_f^i(A)-r_i(A)\,\kappa_f^{i_0(A)}(A)\Big)\,\omega_f(x)=0,
\end{equation}
i.e.\ $h_i(A)^\top\omega(x)=0$ where $h_{i,f}(A):=\kappa_f^i(A)-r_i(A)\,\kappa_f^{i_0(A)}(A)$.

\begin{theorem}[Global identification of the general CA-RRM]\label{thm:global_id_gen}
Consider the CA-RRM with affine indices $u_f(A)=\alpha_f+\beta_f^\top\psi(A)$, $f\in\mathcal{F}$.
Assume $P(i\mid A)\in(0,1)$ for all $i\in A$ and all menus in the support (Assumption~\ref{ass:interior}), and impose the normalization \eqref{eq:norm_gen}. Note that under correct specification, this interior-probability condition implies that every alternative $i\in A$ is recommended by at least one active rule (i.e., $\mathcal{I}_i(A)\neq\emptyset$), since \eqref{eq:P_softmax_gen} assigns zero probability to uncovered alternatives.
Let $d=\dim(\psi)$. Suppose there exist $(d+1)$ feature values $x^{(0)},\ldots,x^{(d)}$ such that:

\begin{enumerate}[label=(G\arabic*)]
\item \textbf{Within-feature variation in decisive-side patterns:} For each $k\in\{0,\ldots,d\}$, stack the row vectors $h_i(A)^\top$ across all multi-sided menus $A$ with $\psi(A)=x^{(k)}$ and all non-reference alternatives $i\in A\setminus\{i_0(A)\}$ into a matrix
\[
H^{(k)}:=\begin{pmatrix}
h_{i_1}(A_1)^\top\\ \vdots\\ h_{i_M}(A_M)^\top
\end{pmatrix}.
\]
Then $\mathrm{rank}(H^{(k)})=|\mathcal{F}|-1$.
\item \textbf{Affine richness of feature values:} The $(d+1)\times(d+1)$ matrix
$X:=\begin{pmatrix}1 & (x^{(0)})^\top\\ \vdots & \vdots\\ 1 & (x^{(d)})^\top\end{pmatrix}$
has full rank.
\end{enumerate}

Then the full parameter vector $\{(\alpha_f,\beta_f)\}_{f\in\mathcal{F}}$ is \emph{globally identified} from the population objects $\{P(i\mid A),\,\kappa^i(A),\,\psi(A)\}$, up to the normalization.
\end{theorem}

\begin{proof}
The proof parallels that of Theorem~\ref{thm:fixed_x_id} and Corollary~\ref{cor:global_id_from_fixedx} in the main text, extended to the multinomial setting.

\textbf{Step 1 (from choice probabilities to linear equations).}
Fix a multi-sided menu $A$ and write $x=\psi(A)$. From \eqref{eq:P_softmax_gen}, the choice probability ratio for any non-reference alternative $i$ is
\[
r_i(A)=\frac{P(i\mid A)}{P(i_0(A)\mid A)}=\frac{\sum_{f}\kappa_f^i(A)\,\omega_f(x)}{\sum_{f}\kappa_f^{i_0(A)}(A)\,\omega_f(x)}.
\]
Cross-multiplying yields the linear restriction \eqref{eq:linear_restriction_gen}: $h_i(A)^\top\omega(x)=0$, where the coefficients $h_{i,f}(A)=\kappa_f^i(A)-r_i(A)\kappa_f^{i_0(A)}(A)$ are determined by the decisive-side indicators and the observed odds ratio.

\textbf{Step 2 (identifying $\omega(x)$ at a fixed feature value up to scale).}
Fix $k\in\{0,\ldots,d\}$. All menus $A$ with $\psi(A)=x^{(k)}$ share the same unknown weight vector $\omega(x^{(k)})$, while the decisive-side indicators may vary across menus even at the same feature value. Stacking all rows $h_i(A)^\top$ from these menus into $H^{(k)}$ gives $H^{(k)}\omega(x^{(k)})=0$. By (G1), $\mathrm{rank}(H^{(k)})=|\mathcal{F}|-1$, so the nullspace is one-dimensional. Since all entries of $\omega(x^{(k)})$ are positive, the weight vector is identified up to a positive scalar.

\textbf{Step 3 (pinning down scale using the normalization).}
Under \eqref{eq:norm_gen}, $\omega_{f_0}(x)=\exp(0+\mathbf{0}^\top x)=1$ for all $x$. Imposing $\omega_{f_0}(x^{(k)})=1$ uniquely pins the scalar, so $\omega(x^{(k)})$ is uniquely identified for each $k$.

\textbf{Step 4 (recovering $(\alpha_f,\beta_f)$ from identified weights).}
Fix any rule $f$. From $\log\omega_f(x^{(k)})=\alpha_f+\beta_f^\top x^{(k)}$ for $k=0,\ldots,d$, we obtain the linear system $(\log\omega_f(x^{(0)}),\ldots,\log\omega_f(x^{(d)}))^\top = X(\alpha_f,\beta_f^\top)^\top$. By (G2), $X$ is invertible, so $(\alpha_f,\beta_f)$ is uniquely determined. Since $f$ was arbitrary, global identification holds.
\end{proof}

\begin{remark}
The binary application in the main text (Corollary~\ref{cor:global_id_from_fixedx}) is the special case with $|A|=2$, where each menu produces a single row $h(A)^\top$ and the odds ratio reduces to $p(A)/(1-p(A))$.
\end{remark}

\subsection{Local identification theorem}

\begin{remark}
Theorem~\ref{thm:global_id_gen} above establishes global identification under conditions (G1)--(G2), which require that multiple menus share the same feature value. The following local result provides an alternative characterization based on the Jacobian of the identification map. It is particularly relevant when the gate operates on a high-dimensional feature map (such as the raw menu encoding $\psi(A)\in\mathbb{R}^{4M}$) for which the replication condition in (G1) may not hold. See Appendix~\ref{app:jacobian_details} for further discussion.
\end{remark}

\begin{theorem}[Local identification of the general CA-RRM]\label{thm:local_id_gen}
Impose the normalization \eqref{eq:norm_gen} and let $\vartheta^\star\in\mathbb{R}^K$ be any parameter vector.
Define the Jacobian
\[
J(\vartheta^\star)\;:=\;\frac{\partial H(\vartheta)}{\partial\vartheta^\top}\bigg|_{\vartheta=\vartheta^\star}\;\in\;\mathbb{R}^{D\times K}.
\]
\begin{enumerate}[label=(\alph*),leftmargin=*]
\item If $J(\vartheta^\star)$ has full column rank $K$, then $\vartheta^\star$ is locally identified from $\{p(i\mid A)\}_{A\in\mathcal{M}_{\mathrm{ms}}}$: there exists a neighborhood $U\ni\vartheta^\star$ such that $H(\vartheta)=H(\vartheta^\star)$ and $\vartheta\in U$ imply $\vartheta=\vartheta^\star$.
\item If $\mathrm{rank}(J(\vartheta^\star))=r<K$, then the null space $\mathcal{N}(J(\vartheta^\star))$ has dimension $K-r$, and $H$ is locally constant to first order along these $K-r$ linearly independent directions. If, additionally, the rank of $J(\vartheta)$ equals $r$ on a neighborhood of $\vartheta^\star$, then by the constant-rank theorem the level set $\{\vartheta: H(\vartheta)=H(\vartheta^\star)\}$ is locally a smooth embedded submanifold of dimension $K-r$ with tangent space $\mathcal{N}(J(\vartheta^\star))$.
\end{enumerate}
\end{theorem}

\begin{proof}
Each score $u_f(A;\vartheta)$ is affine in $\vartheta$, and the log-sum-exp function is $C^\infty$ on $\mathbb{R}^n$ for any $n\ge 1$. Therefore $\delta_i(A;\vartheta)$ in \eqref{eq:logratio_gen}, being a difference of two log-sum-exp functions composed with affine maps, is $C^\infty$ in $\vartheta$. The identification map $H$ is thus a $C^\infty$ map from $\mathbb{R}^K$ to $\mathbb{R}^D$.

For part~(a): if $J(\vartheta^\star)\in\mathbb{R}^{D\times K}$ has rank $K$, there exist $K$ components of $H$ whose $K\times K$ sub-Jacobian is nonsingular (select any $K$ linearly independent rows of $J(\vartheta^\star)$). By the inverse function theorem applied to this square submap, the selected components of $H$ are locally injective at $\vartheta^\star$. Since $H(\vartheta)=H(\vartheta^\star)$ implies equality of the selected subvector, the full map $H$ is locally injective.

For part~(b): the first-order statement follows directly from $\mathrm{rank}(J(\vartheta^\star))=r$: by the rank-nullity theorem, $\mathcal{N}(J(\vartheta^\star))$ has dimension $K-r$, and for any $v\in\mathcal{N}(J(\vartheta^\star))$, the directional derivative $DH(\vartheta^\star)\cdot v=0$. If, additionally, $\mathrm{rank}(J(\vartheta))=r$ on a neighborhood of $\vartheta^\star$, the constant-rank theorem yields the stronger geometric conclusion that the level set is locally a smooth $(K-r)$-dimensional submanifold. Without this local constancy assumption, the first-order characterization still holds but the level set may have dimension less than $K-r$.
\end{proof}

\subsection{Closed-form Jacobian}\label{app:jacobian}

We derive the Jacobian entries in closed form. For each active rule $f\in\mathcal{F}_A$, define the \emph{active-set softmax share} (i.e., the share of rule $f$ in the softmax over all active rules, not conditional on alternative)
\begin{equation}\label{eq:softmax_share_gen}
s_f(A;\vartheta)\;:=\;\frac{\exp(u_f(A;\vartheta))}{\sum_{g\in\mathcal{F}_A}\exp(u_g(A;\vartheta))},
\end{equation}
and for each alternative $i\in A$, write $S_i(A;\vartheta):=\sum_{f\in\mathcal{I}_i(A)}s_f(A;\vartheta)$ for the total share allocated to alternative~$i$. Note that $P(i\mid A;\vartheta)=S_i(A;\vartheta)$.

\begin{proposition}[Jacobian entries]\label{prop:jacobian_gen}
For any multi-sided menu $A\in\mathcal{M}_{\mathrm{ms}}$, any non-reference alternative $i\in A\setminus\{i_0(A)\}$, and any free rule $f\neq f_0$:
\begin{align}
\frac{\partial\delta_i(A)}{\partial\alpha_f}
&\;=\;
\frac{s_f(A)\,\Ind\{f\in\mathcal{I}_i(A)\}}{S_i(A)}
\;-\;
\frac{s_f(A)\,\Ind\{f\in\mathcal{I}_{i_0}(A)\}}{S_{i_0}(A)},
\label{eq:jac_alpha_gen}\\[4pt]
\frac{\partial\delta_i(A)}{\partial\beta_f}
&\;=\;
\left(
\frac{s_f(A)\,\Ind\{f\in\mathcal{I}_i(A)\}}{S_i(A)}
\;-\;
\frac{s_f(A)\,\Ind\{f\in\mathcal{I}_{i_0}(A)\}}{S_{i_0}(A)}
\right)\psi(A)^\top.
\label{eq:jac_beta_gen}
\end{align}
If $f\notin\mathcal{F}_A$ (i.e., $f$ is inactive at $A$), then $s_f(A)=0$ and both derivatives vanish.
\end{proposition}

\begin{proof}
From \eqref{eq:logratio_gen}, $\delta_i(A;\vartheta)=\log N_i(A) - \log N_{i_0}(A)$, where $N_j(A):=\sum_{g\in\mathcal{I}_j(A)}\exp(u_g(A))$.
For $f\in\mathcal{I}_i(A)$:
\[
\frac{\partial\log N_i(A)}{\partial\alpha_f}
=\frac{\exp(u_f(A))}{N_i(A)}
=\frac{s_f(A)\cdot Z(A)}{N_i(A)}
=\frac{s_f(A)}{S_i(A)},
\]
where $Z(A):=\sum_{g\in\mathcal{F}_A}\exp(u_g(A))$ and $S_i(A)=N_i(A)/Z(A)$.
For $f\notin\mathcal{I}_i(A)$, the derivative is zero.
Combining the contributions from the two log terms and noting $\partial u_f/\partial\beta_f=\psi(A)^\top$ gives \eqref{eq:jac_alpha_gen}--\eqref{eq:jac_beta_gen}.
\end{proof}

\paragraph{Block structure.}
The Jacobian $J(\vartheta^\star)\in\mathbb{R}^{D\times K}$ has $(F-1)$ column blocks, one per free rule. Each block has $1+d$ columns (intercept $\alpha_f$ plus slopes $\beta_f$). The derivative for rule $f$ at menu $A$ is nonzero only if $f$ is active at $A$ and belongs to $\mathcal{I}_i(A)\cup\mathcal{I}_{i_0}(A)$, i.e., $f$ recommends either alternative $i$ or the reference $i_0(A)$. Variation in which rules fall on which ``side'' across menus-i.e., rule switching-produces sign changes in the Jacobian entries, which tends to increase rank.

\begin{remark}
The binary application in the main text (Corollary~\ref{cor:global_id_from_fixedx}) is the special case of Theorem~\ref{thm:global_id_gen} obtained by setting $|A|=2$, where log-ratios reduce to log-odds and the Jacobian entries simplify to $\partial\delta(A)/\partial\alpha_f = s_{L,f}(A)-s_{R,f}(A)$. Three sources of variation contribute to rank: side-switching (rules recommend different alternatives at different menus), feature variation (richer $\psi(A)$ generates more independent rows), and compositional variation (changes in which rules are active). The framework applies to any environment with a finite library of deterministic procedures and an exogenous activity criterion.
\end{remark}

\section{Proof of Theorem~\ref{thm:fixed_x_id} (fixed-feature identification)}\label{app:proof_fixed_x}

\begin{proof}
Fix $x\in\mathbb R^d$ and let $\mathcal A_2(x):=\{A\in\mathcal A_2:\psi(A)=x\}$.
For each menu $A\in\mathcal A_2(x)$, define the odds ratio $r(A):=p(A)/(1-p(A))$ and the
column vector $h(A)\in\mathbb R^{|\mathcal F|}$ with entries
$h_f(A)=\kappa_f^L(A)-r(A)\kappa_f^R(A)$.
Denote the true weights by $\omega_f(x)=\exp(\alpha_f+\beta_f^\top x)>0$.

\textbf{Step 1 (from choice probabilities to linear restrictions).}
Under the CA-RRM the model-implied left-choice probability at menu $A$ is
\[
p(A)=\frac{\sum_{f}\kappa_f^L(A)\,\omega_f(x)}
{\sum_{f}\kappa_f^L(A)\,\omega_f(x)+\sum_{f}\kappa_f^R(A)\,\omega_f(x)}.
\]
Cross-multiplying $p(A)\sum_f\kappa_f^R(A)\omega_f(x)=(1-p(A))\sum_f\kappa_f^L(A)\omega_f(x)$ and dividing by $(1-p(A))$ gives
\[
\sum_{f\in\mathcal F}\bigl(\kappa_f^L(A)-r(A)\,\kappa_f^R(A)\bigr)\,\omega_f(x)
\;=\;h(A)^\top\omega(x)\;=\;0.
\]
Since this holds for every $A\in\mathcal A_2(x)$, any vector in the cellwise identified
set $\Omega(x)$ lies in the null space of the matrix $H(x)$ formed by stacking the rows
$\{h(A)^\top:A\in\mathcal A_2(x)\}$.

\textbf{Step 2 (sufficiency: $\rank H(x)=|\mathcal F|-1$ $\Rightarrow$ single ray).}
Suppose $\rank H(x)=|\mathcal F|-1$. Then $\ker H(x)$ is one-dimensional.
Let $v$ be a nonzero vector spanning this kernel.
Because the true weights $\omega(x)\in\mathbb R_{++}^{|\mathcal F|}$ lie in
$\Omega(x)\subseteq\ker H(x)$, we have $\omega(x)=c\,v$ for some scalar $c\neq 0$.
Since $\omega(x)$ is strictly positive, $v$ must be a scalar multiple of a strictly
positive vector. Hence every $\tilde\omega\in\Omega(x)$ satisfies
$\tilde\omega=c'\,v$ for some $c'>0$, i.e.\
$\Omega(x)=\{c'\,\omega^\star(x):c'>0\}$ is a single positive ray.

\textbf{Step 3 (necessity: single ray $\Rightarrow$ $\rank H(x)=|\mathcal F|-1$).}
We prove the contrapositive. Suppose $\rank H(x)<|\mathcal F|-1$.
Then $\dim\ker H(x)\ge 2$.
Pick two linearly independent vectors $v_1,v_2\in\ker H(x)$ with $v_1=\omega(x)\in\mathbb R_{++}^{|\mathcal F|}$.
For all sufficiently small $\varepsilon>0$, the vector $v_1+\varepsilon\,v_2$ is also
strictly positive (since $v_1$ has strictly positive entries).
Moreover $v_1+\varepsilon v_2\in\ker H(x)$ by linearity.
Therefore $v_1+\varepsilon v_2\in\Omega(x)$, and it is not proportional to $v_1$
(since $v_1$ and $v_2$ are linearly independent), so $\Omega(x)$ is not a single ray.

\textbf{Step 4 (point identification under normalization).}
Normalizing $\omega_{f_0}(x)=1$ restricts $\Omega(x)$ to the hyperplane $\{\omega:\omega_{f_0}=1\}$.
A single positive ray intersects this hyperplane in a unique point, so the normalization
converts ray identification into point identification.
\end{proof}

\section{Proof of Corollary~\ref{cor:global_id_from_fixedx} (detailed version)}\label{app:proof_global_id}

\begin{proof}
Fix any menu $A$ and write $x=\psi(A)$. Define the positive rule weights
\[
\omega_f(x):=\exp(\alpha_f+\beta_f^\top x),\qquad f\in\mathcal F,
\]
and let $\kappa_f^L(A)=\mathbf 1\{f\in \mathcal I_L(A)\}$ and $\kappa_f^R(A)=\mathbf 1\{f\in \mathcal I_R(A)\}$ denote whether rule $f$
is decisive for the left or right option at $A$.

\textbf{Step 1 (from choice probabilities to linear equations).}
Under the CA-RRM, the model-implied probability of choosing left at menu $A$ is
\begin{equation}\label{eq:prob_rewrite}
p(A)
=
\frac{\sum_{f\in\mathcal F}\kappa_f^L(A)\,\omega_f(x)}
{\sum_{f\in\mathcal F}\kappa_f^L(A)\,\omega_f(x)+\sum_{f\in\mathcal F}\kappa_f^R(A)\,\omega_f(x)}.
\end{equation}
Because $p(A)\in(0,1)$, the odds ratio
\[
r(A):=\frac{p(A)}{1-p(A)}\in(0,\infty)
\]
is well-defined. Using \eqref{eq:prob_rewrite}, the odds can be written as
\[
r(A)=
\frac{\sum_{f}\kappa_f^L(A)\,\omega_f(x)}
{\sum_{f}\kappa_f^R(A)\,\omega_f(x)}.
\]
Cross-multiplying yields
\[
\sum_{f}\kappa_f^L(A)\,\omega_f(x) \;=\; r(A)\sum_{f}\kappa_f^R(A)\,\omega_f(x),
\]
and bringing all terms to the left-hand side gives the linear restriction
\begin{equation}\label{eq:linear_restriction_proof}
\sum_{f\in\mathcal F}\Big(\kappa_f^L(A)-r(A)\kappa_f^R(A)\Big)\,\omega_f(x)=0.
\end{equation}
Define the row vector $h(A)\in\mathbb R^{|\mathcal F|}$ by
\[
h_f(A):=\kappa_f^L(A)-r(A)\kappa_f^R(A),
\]
so \eqref{eq:linear_restriction_proof} can be written compactly as
\begin{equation}\label{eq:h_w}
h(A)^\top \omega(x)=0,
\end{equation}
where $\omega(x)=(\omega_f(x))_{f\in\mathcal F}$.

\textbf{Step 2 (identifying $\omega(x)$ at a fixed feature value up to scale).}
Fix $k\in\{0,\ldots,d\}$ and consider the menus $A_{k1},\ldots,A_{kM_k}$ that satisfy $\psi(A_{km})=x^{(k)}$.
Crucially, the gate depends on $A$ only through $x=\psi(A)$, so all these menus share the \emph{same} unknown weight vector
$\omega(x^{(k)})$, while the decisive-side indicators $(\kappa^L,\kappa^R)$ (and hence $h(A)$) may vary across menus even when
$x$ is fixed. Applying \eqref{eq:h_w} to each menu and stacking the resulting equations gives
\[
\begin{pmatrix}
h(A_{k1})^\top\\ \vdots\\ h(A_{kM_k})^\top
\end{pmatrix}
\omega(x^{(k)})=0,
\qquad\text{i.e.}\qquad
H^{(k)}\omega(x^{(k)})=0.
\]
By (G1), $\mathrm{rank}(H^{(k)})=|\mathcal F|-1$, so the nullspace of $H^{(k)}$ is one-dimensional: there exists a nonzero
vector $\bar \omega^{(k)}$ such that every solution to $H^{(k)}\omega=0$ is of the form $\omega=c\,\bar \omega^{(k)}$ for some scalar $c$.
The true weight vector $\omega(x^{(k)})$ is strictly positive and lies in this nullspace, so $\bar\omega^{(k)}$ may be oriented to be strictly positive (set $\bar\omega^{(k)}:=\omega(x^{(k)})$). Every admissible solution---i.e., every $\omega\in\mathbb R_{++}^{|\mathcal F|}$ with $H^{(k)}\omega=0$---is then exactly a positive multiple $c\,\bar\omega^{(k)}$ with $c>0$. Thus, the menus with $\psi(A)=x^{(k)}$
identify $\omega(x^{(k)})$ \emph{up to a positive multiplicative scale factor}.

\textbf{Step 3 (pinning down scale using the normalization).}
Under the normalization $(\alpha_{f_0},\beta_{f_0})=(0,0)$ we have, for all $x$,
\[
\omega_{f_0}(x)=\exp(0+0^\top x)=1.
\]
In particular, $\omega_{f_0}(x^{(k)})=1$. Since every solution to $H^{(k)}\omega=0$ equals $c\,\bar \omega^{(k)}$, imposing the single
restriction $\omega_{f_0}=1$ pins down $c$ uniquely:
\[
1=\omega_{f_0}(x^{(k)})=c\,\bar \omega^{(k)}_{f_0}
\quad\Rightarrow\quad
c=\frac{1}{\bar \omega^{(k)}_{f_0}}.
\]
Hence there is a unique solution to $H^{(k)}\omega=0$ that satisfies $\omega_{f_0}=1$, and therefore $\omega(x^{(k)})$ is uniquely
identified for each $k$.

\textbf{Step 4 (recovering $(\alpha_f,\beta_f)$ from identified weights across $k$).}
Fix any rule $f\in\mathcal F$. From the definition of $\omega_f(\cdot)$,
\[
\log \omega_f(x^{(k)})=\alpha_f+\beta_f^\top x^{(k)},\qquad k=0,\ldots,d.
\]
Stacking these $(d+1)$ equalities yields the linear system
\[
\begin{pmatrix}
\log \omega_f(x^{(0)})\\ \vdots\\ \log \omega_f(x^{(d)})
\end{pmatrix}
=
\begin{pmatrix}
1 & (x^{(0)})^\top\\
\vdots & \vdots\\
1 & (x^{(d)})^\top
\end{pmatrix}
\begin{pmatrix}
\alpha_f\\ \beta_f
\end{pmatrix}
=
X\begin{pmatrix}\alpha_f\\ \beta_f\end{pmatrix}.
\]
Assumption (G2) states that the $(d+1)$ feature vectors $x^{(0)},\ldots,x^{(d)}$ are \emph{affinely independent},
equivalently that the augmented regressor matrix $X$ has full rank; hence $X$ is invertible. Therefore
$(\alpha_f,\beta_f)$ is uniquely determined. Since $f$ was arbitrary, this holds for all $f\in\mathcal F$, establishing
global identification (up to the imposed normalization).
\end{proof}

\section{Two-Step Estimator: Construction and Regularity Conditions}\label{app:two_step_details}

This section provides the detailed construction of the two-step estimator summarized in Section~\ref{subsec:estimation_inference}, including the sampling scheme, the estimator formulas, the regularity conditions, population functionals, and the overidentification test.

\paragraph{Sampling scheme.}
Let $\mathcal M=\{A_1,\ldots,A_M\}$ be a finite collection of menus.
For each menu $A_m$, we observe $n_m$ independent binary choices
$Y_{mi}\in\{0,1\}$, $i=1,\ldots,n_m$, with
$\Pr(Y_{mi}=1\mid A_m)=p(A_m)\in(0,1)$.
For distinct menus, the menu-specific samples are mutually independent.
Write
$\hat p_m:=n_m^{-1}\sum_{i=1}^{n_m}Y_{mi}$ and
$\hat r_m:=\hat p_m/(1-\hat p_m)$,
where, for numerical stability only, $\hat p_m$ may be trimmed to
$[\varepsilon_n,1-\varepsilon_n]$ with $\varepsilon_n\downarrow 0$ slowly enough
that trimming is asymptotically inactive.
We study the asymptotic regime
$N:=\min_{1\le m\le M} n_m \to\infty$,
$n_m/N\to \pi_m\in(0,\infty)$ for each~$m$.

\paragraph{Identifying feature values.}
Let $x^{(1)},\ldots,x^{(K)}$ be distinct feature values with $K\ge d+1$, and define
the associated menu cells
$\mathcal M_k:=\{m:\psi(A_m)=x^{(k)}\}$.
For each $k$, let $M_k:=|\mathcal M_k|$ and assume $M_k\ge |\mathcal F|-1$.
For $m\in\mathcal M_k$, define
$h_f(A_m):=\kappa_f^L(A_m)-r(A_m)\kappa_f^R(A_m)$ and
$\hat h_f(A_m):=\kappa_f^L(A_m)-\hat r_m\,\kappa_f^R(A_m)$.
Let $H^{(k)}$ and $\hat H^{(k)}$ denote the $M_k\times |\mathcal F|$ matrices stacking
the row vectors $h(A_m)^\top$ and $\hat h(A_m)^\top$ over $m\in\mathcal M_k$.
Fix the baseline rule $f_0$ and partition
$H^{(k)}=[h_0^{(k)}\ \ H_-^{(k)}]$,
$\hat H^{(k)}=[\hat h_0^{(k)}\ \ \hat H_-^{(k)}]$.

\paragraph{First-stage estimator: cellwise rule weights.}
Under Theorem~\ref{thm:fixed_x_id}, the normalized weight vector at $x^{(k)}$
is uniquely determined by $\omega_{f_0}(x^{(k)})=1$ and
$H^{(k)}\omega(x^{(k)})=0$.
Writing $\omega(x^{(k)})=(1,\eta^{(k)\top})^\top$,
the population vector $\eta^{(k)}$ solves
$H_-^{(k)}\eta^{(k)}=-h_0^{(k)}$.
The sample estimator is
\begin{equation}\label{eq:eta_hat}
\hat\eta^{(k)}:=
-\Big(\hat H_-^{(k)\top}\hat H_-^{(k)}\Big)^{-1}\hat H_-^{(k)\top}\hat h_0^{(k)},
\qquad
\hat\omega(x^{(k)}):=(1,\hat\eta^{(k)\top})^\top.
\end{equation}

\paragraph{Second-stage estimator: affine gate parameters.}
For each non-baseline rule $f\neq f_0$, define
$\gamma_f:=(\alpha_f,\beta_f^\top)^\top\in\mathbb R^{d+1}$ and
$z_k:=(1,x^{(k)\top})^\top$.
The affine gate implies
$\log \omega_f(x^{(k)}) = z_k^\top \gamma_f$, $k=1,\ldots,K$.
Define $\hat y_f^{(k)}:=\log \hat\omega_f(x^{(k)})$ and let
$X_K$ stack the vectors $z_k^\top$ over $k$.
The minimum-distance estimator is
\begin{equation}\label{eq:gamma_hat}
\hat\gamma_f:=
(X_K^\top W_f X_K)^{-1}X_K^\top W_f \hat y_f.
\end{equation}

\paragraph{Population functionals.}
The two-step estimator identifies economically meaningful functionals of the
gate, including average latent gate mass $\bar q_f(\theta):=M^{-1}\sum_{m=1}^M q_f(A_m;\theta)$,
average effective active mass $\bar a_f(\theta):=M^{-1}\sum_{m=1}^M q_f(A_m;\theta)\,A_f(A_m)$,
and average conditional-on-activity responsibility
$w_f(\theta):=M^{-1}\sum_{m=1}^M
q_f(A_m;\theta)\,A_f(A_m)/\sum_{g\in\mathcal F} q_g(A_m;\theta)\,A_g(A_m)$.

\begin{assumption}\label{ass:estimation}
The following hold:
\begin{enumerate}[label=(E\arabic*)]
\item \textbf{Interior probabilities.}
There exists $\underline p\in(0,1/2)$ such that
$p(A_m)\in[\underline p,1-\underline p]$ for all $m$.
\item \textbf{Cellwise identification.}
For each $k=1,\ldots,K$, the matrix $H^{(k)}$ has rank $|\mathcal F|-1$.
\item \textbf{Affine richness.}
The matrix $X_K$ has full column rank $d+1$.
\item \textbf{Interior weights.}
For every $k$ and every $f\neq f_0$,
$\omega_f(x^{(k)})\ge \underline\omega$ for some $\underline\omega>0$.
\item \textbf{Stable sampling shares.}
$n_m/N\to\pi_m\in(0,\infty)$ for each menu $m$.
\end{enumerate}
\end{assumption}

\paragraph{Efficient weighting and feasible inference.}
A natural choice is feasible GLS, using a consistent estimator $\hat V_{y,f}$ of the
first-stage covariance matrix and setting $W_f=\hat V_{y,f}^{-1}$.
In practice, a cluster bootstrap at the menu level directly yields
standard errors for both the affine coefficients and the decomposition functionals.

\paragraph{Overidentification test of the affine gate.}
When $K>d+1$, the affine gate restriction is overidentified. For each rule
$f\neq f_0$, define the second-stage residual vector
$\hat u_f:=\hat y_f-X_K\hat\gamma_f$.
A natural minimum-distance specification test is
$J_f:=N\,\hat u_f^\top \hat V_{y,f}^{-1}\hat u_f$,
which under efficient weighting and the null that
$\log\omega_f(x)$ is affine in $x$ has the asymptotic reference
distribution $J_f\dto\chi^2_{K-d-1}$.

\section{Proof of Theorem~\ref{thm:two_step}}\label{app:proof_two_step}

\begin{proof}
Let $J:=|\mathcal F|$. Write $p_m:=p(A_m)$, $r_m:=p_m/(1-p_m)$, $\hat p_m:=n_m^{-1}\sum_{i=1}^{n_m}Y_{mi}$, $\hat r_m:=\hat p_m/(1-\hat p_m)$, and $\omega^{(k)}:=(1,\eta^{(k)\top})^\top$ where $\eta^{(k)}\in\mathbb R^{J-1}$ collects the non-baseline weights at cell $k$.

\noindent\textbf{Part (i).}
For each menu $m$, the multivariate CLT and the delta method applied to $g(u)=u/(1-u)$ give
\begin{equation}\label{eq:odds_clt}
\sqrt N(\hat r-r)\dto \mathcal N(0,\Sigma_r),
\qquad
\Sigma_r=\diag\!\Big(\frac{p_m}{\pi_m(1-p_m)^3}\Big).
\end{equation}
Since $H^{(k)}$ depends on the data only through the odds ratios and does so affinely, $\sqrt N\,\mathrm{vec}(\hat H^{(k)}-H^{(k)})\dto \mathcal N(0,\Sigma_{H}^{(k)})$. Under~(E2), $H_-^{(k)}$ has full column rank $J-1$, so the least-squares map $\phi_k(H):=-(H_-^\top H_-)^{-1}H_-^\top h_0$ is smooth, and the delta method yields $\sqrt N(\hat\eta^{(k)}-\eta^{(k)})\dto\mathcal N(0,V_\eta^{(k)})$. Embedding via $\hat\omega^{(k)}=(1,\hat\eta^{(k)\top})^\top$ gives $\sqrt N(\hat\omega^{(k)}-\omega^{(k)})\dto\mathcal N(0,V_\omega^{(k)})$, where $V_\omega^{(k)}$ has a zero row/column at $f_0$ (reflecting the normalization). Under~(E4), the log transformation is smooth and a further delta-method step gives $\sqrt N(\hat y_f^{(k)}-y_f^{(k)})\dto\mathcal N(0,V_{y,f}^{(k)})$ with $V_{y,f}^{(k)}=\omega_f(x^{(k)})^{-2}\,e_f^\top V_\omega^{(k)} e_f$. Cross-cell independence (disjoint menu partitions) yields $\sqrt N(\hat y_f-y_f)\dto \mathcal N(0,V_{y,f})$ with $V_{y,f}=\operatorname{diag}(V_{y,f}^{(1)},\ldots,V_{y,f}^{(K)})$.

\noindent\textbf{Part (ii).}
By construction $y_f=X_K\gamma_f$. Let $A_f:=(X_K^\top W_f X_K)^{-1}X_K^\top W_f$ (well-defined by (E3)). Then $\hat\gamma_f-\gamma_f=A_f(\hat y_f-y_f)$, so $\sqrt N(\hat\gamma_f-\gamma_f)\dto\mathcal N(0,A_f V_{y,f} A_f^\top)$.

\noindent\textbf{Part (iii).}
Standard delta method applied to $\tau(\hat\theta)$.
\end{proof}

\section{Perceived-Lottery Specifications}\label{app:perceived_lottery_specs}

This section provides the formal perceived-lottery map $\pi_f$ for each rule in the library (Table~\ref{tab:rules_summary}), along with the notation needed to state them.

\paragraph{Notation.}
Let $\text{Supp}(L)$ denote the finite support of lottery $L$, $x_{\min}(L):=\min\text{Supp}(L)$, $x_{\max}(L):=\max\text{Supp}(L)$.
Write $\delta(z)$ for the degenerate lottery yielding $z$ with probability one. Let $z_{\mathrm{mode}}(L)$ denote the (highest) modal payoff of $L$ (ties broken upward).
Given outcomes $x,y\in\mathbb R$, the normalized contrast is
$\mathrm{contr}(x,y):=|x-y|/(|x|+|y|+1)$.
The constant $+1$ is a fixed regularization choice (not a tunable parameter). Because the contrast is not homogeneous of degree zero, the effective rule recommendations depend on the outcome scale of the data. The rules remain parameter-free in the sense that no parameter is estimated from data, but they are not invariant to affine transformations of outcomes. This scale dependence is a limitation of the current formulation.

\paragraph{First-Order Stochastic Dominance.}
\begin{definition}[First-Order Stochastic Dominance]\label{def:FSD}
Let $L^1,L^2\in X$. Let $\mathrm{supp}(L^1)\cup \mathrm{supp}(L^2)=\{z_1<\dots<z_M\}$ be the ordered union of their supports. For $\ell\in\{1,2\}$, let $\pi^\ell_k := \Pr_{L^\ell}(z_k)$ with the convention $\pi^\ell_k=0$ when $z_k\notin\mathrm{supp}(L^\ell)$, and define the right-tail cumulative sums $S^\ell_i := \sum_{k=i}^{M}\pi^\ell_k$, $i=1,\dots,M$. We say that $L^1\ge_{\mathrm{FSD}} L^2$ if and only if $S^1_i\ge S^2_i$ for all $i=1,\dots,M$. We denote the strict part by $>_{\mathrm{FSD}}$.
\end{definition}

\paragraph{Auxiliary definitions.}
For the salience rules, define the set of extreme pairings $\mathcal P(A):=\{(a,b): a\in\{x_{\min}(L^1),x_{\max}(L^1)\},\ b\in\{x_{\min}(L^2),x_{\max}(L^2)\}\}$, which contains at most four outcome pairs, with salience scores $c_k:=\mathrm{contr}(a_k,b_k)$.
For the regret rules, define a canonical state space as the set of all outcome pairs $(x_i,y_j)$ with $x_i\in\mathrm{Supp}(L^1)$, $y_j\in\mathrm{Supp}(L^2)$, each occurring with probability $p_i^{L^1}\cdot p_j^{L^2}$, with regrets $d^{L}_{ij}:=[y_j-x_i]_+$ and $d^{R}_{ij}:=[x_i-y_j]_+$.
For the disappointment rules, define the set of downside payoffs $Z^{-}(L):=\{z\in\text{Supp}(L): z<z_{\mathrm{mode}}(L)\}$ and the disappointment contrasts $s(z):=\mathrm{contr}(z_{\mathrm{mode}}(L),z)$ for $z\in Z^{-}(L)$.

\paragraph{Extremum rules.}
$\pi_{\textsc{MMn}}(L^i;L^j):=\delta(x_{\min}(L^i))$; \ $\pi_{\textsc{MMx}}(L^i;L^j):=\delta(x_{\max}(L^i))$; \ $\pi_{\textsc{MMa}}(L^i;L^j):=\delta\!\big(\tfrac12(x_{\min}(L^i)+x_{\max}(L^i))\big)$; \ $\pi_{\textsc{MAP}}(L^i;L^j):=\delta(z_{\mathrm{mode}}(L^i))$.

\paragraph{Salience rules.}
For \textsc{SAL}, let $k^\star\in\arg\max_k c_k$ (ties broken deterministically): $\pi_{\textsc{SAL}}(L^1;L^2):=\delta(a_{k^\star})$, $\pi_{\textsc{SAL}}(L^2;L^1):=\delta(b_{k^\star})$.
For \textsc{SAL2}, let $c_{(1)}\ge c_{(2)}$ be the two largest salience scores. When $c_{(1)}=c_{(2)}$ the rule is inactive; otherwise let $k^{(2)}$ attain $c_{(2)}$ and set $\pi_{\textsc{SAL2}}(L^1;L^2):=\delta(a_{k^{(2)}})$, $\pi_{\textsc{SAL2}}(L^2;L^1):=\delta(b_{k^{(2)}})$.

\paragraph{Regret rules.}
Following \citet{loomes_sugden1982}:
$\pi_{\textsc{REG}}(L^1;L^2):=\sum_{i,j} p_i^{L^1}p_j^{L^2}\,\delta(-d^{L}_{ij})$,
$\pi_{\textsc{REG}}(L^2;L^1):=\sum_{i,j} p_i^{L^1}p_j^{L^2}\,\delta(-d^{R}_{ij})$;
active when one negated-regret distribution strictly FSD-dominates the other.
For \textsc{REGmed}: $\pi_{\textsc{REGmed}}(L^1;L^2):=\delta(-R^{L}_{\mathrm{med}})$,
$\pi_{\textsc{REGmed}}(L^2;L^1):=\delta(-R^{R}_{\mathrm{med}})$,
where $R^{L}_{\mathrm{med}}:=\mathrm{wmed}(\{d^{L}_{ij}\},\{p_i^{L^1} p_j^{L^2}\})$ and $R^{R}_{\mathrm{med}}$ is defined symmetrically.

\paragraph{Disappointment rules.}
$D(L):=\max_{z\in Z^{-}(L)} s(z)$ (with $D(L)=0$ if $Z^{-}(L)=\emptyset$);
$\pi_{\textsc{DIS}}(L^i;L^j):=\delta(-D(L^i))$.
For \textsc{DISmed}: if $|Z^{-}(L)|\ge 2$, let $s_{(2)}$ be the second-largest contrast and set $D_{\mathrm{med}}(L):=s_{(2)}$; otherwise $D_{\mathrm{med}}(L):=D(L)$; then $\pi_{\textsc{DISmed}}(L^i;L^j):=\delta(-D_{\mathrm{med}}(L^i))$.

\paragraph{Attention rules.}
Fix $M>0$ so that $-M$ lies strictly below all payoffs in the dataset.
$\pi_{\textsc{A1}}(L^1;L^2)=L^1$, $\pi_{\textsc{A1}}(L^2;L^1)=\delta(-M)$ (\textsc{A1} always recommends $L^1$).
$\pi_{\textsc{A2}}(L^1;L^2)=\delta(-M)$, $\pi_{\textsc{A2}}(L^2;L^1)=L^2$ (\textsc{A2} always recommends $L^2$).
Both are always active; they ensure that at least one rule is decisive at every menu, stabilizing the CA-RRM normalization.

\section{Library Coverage}\label{app:overlap}

\begin{figure}[h!]
    \centering
    \includegraphics[width=0.70\linewidth]{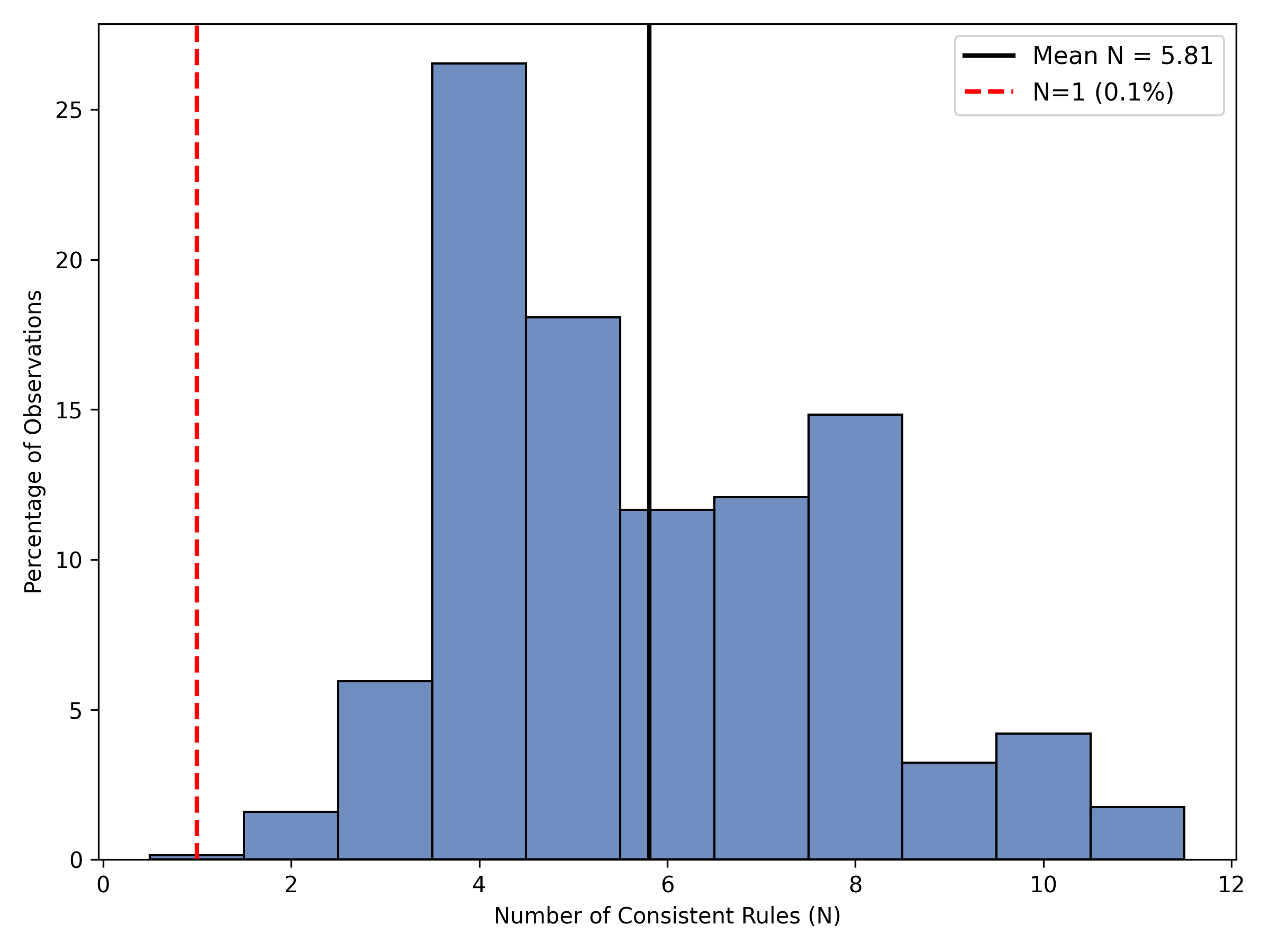}
    \caption{\label{fig:overlap_hist} Distribution of $N_{\mathrm{cons}}(A)$, the number of rules both active and consistent with the majority choice direction at each menu.}
\end{figure}

\section{Cross-Fitted Library Restrictions}\label{app:crossfitted}

\begin{table}[t]
\centering
\caption{\label{tab:crossfitted}Cross-fitted library restrictions: mean test MSE for top-$k$ and top-family selections from the full 12-rule library. Subsets are chosen on training folds only; performance is evaluated out of sample.}
\begin{tabular}{llrrr}
\toprule
Selection type & $k$ & Mean MSE & SE & Retention (\%) \\
\midrule
Full (12 rules) & 12 & 0.01199 & 0.00007 & 100.0 \\
Top-$k$ & 7 & 0.01266 & 0.00008 & 94.4 \\
Top-$k$ & 5 & 0.01388 & 0.00010 & 84.2 \\
Top-$k$ & 3 & 0.01992 & 0.00013 & 33.9 \\
\midrule
Full (5 families) & 5 & 0.01198 & 0.00007 & 100.0 \\
Top-family & 4 & 0.01214 & 0.00007 & 98.6 \\
Top-family & 3 & 0.01324 & 0.00009 & 89.5 \\
Top-family & 2 & 0.01484 & 0.00009 & 76.1 \\
Top-family & 1 & 0.02304 & 0.00013 & 7.6 \\
\bottomrule
\end{tabular}\\[3pt]
{\footnotesize\textit{Note:} Retention~$(\%)=100\times[1-(\text{MSE}^k-\text{MSE}^{\text{full}})/\text{MSE}^{\text{full}}]$: 100\% means no loss; values below 0\% indicate MSE exceeds twice the full-library MSE. Selection is cross-fitted: rule subsets chosen on training data, evaluated on held-out test data.}
\end{table}

\section{Activity Discipline Variation}\label{app:epsilon_fsd}

\begin{table}[t]
\centering
\caption{\label{tab:epsilon_fsd}Activity discipline robustness: activity statistics and out-of-sample MSE under varied dominance thresholds. The baseline ($\varepsilon=0$) corresponds to strict FSD; $\varepsilon>0$ strengthens it; $\varepsilon<0$ removes it (all rules active).}
\begin{tabular}{rrrrr}
\toprule
$\varepsilon$ & Mean active rules & Mean overlap & Mean MSE & SD \\
\midrule
-2.00 & 12.00 & 7.03 & 0.01202 & 0.00053 \\
0.00 & 9.48 & 5.81 & 0.01168 & 0.00049 \\
0.01 & 9.12 & 5.56 & 0.01307 & 0.00060 \\
0.05 & 8.99 & 5.47 & 0.01364 & 0.00064 \\
0.10 & 8.83 & 5.36 & 0.01413 & 0.00066 \\
\bottomrule
\end{tabular}\\[3pt]
{\footnotesize\textit{Note:} For $\varepsilon\ge 0$, $\varepsilon$-FSD requires $S^1(z)\ge S^2(z)+\varepsilon$ for all $z$ (with strict inequality somewhere). For $\varepsilon<0$, all rules are declared active on every menu; recommendations follow standard FOSD where applicable. Mean active rules: average number of decisive rules per menu. Mean overlap: average number of active rules consistent with the majority choice direction. MSE is from the same 50-fold CV protocol as Table~\ref{tab:cv_results}.}
\end{table}

\section{Attention-Rule Removal}\label{app:attention_removal}

Table~\ref{tab:cv_specA} replicates the main cross-validation comparison (Table~\ref{tab:cv_results}) with the rule-gating model restricted to the 10 economic rules (excluding \textsc{A1}/\textsc{A2}).

\begin{table}[h!]
\caption{\label{tab:cv_specA}Robustness: predictive performance after removing attention defaults \textsc{A1}/\textsc{A2}. Rule-gating (10 rules) excludes the two attention defaults; all other models are identical to Table~\ref{tab:cv_results}.}
\centering
\begin{tabular}[t]{lrrrr}
\toprule
Model & Mean MSE & SD & Mean $\mathrm{MSE}_w$ & Best LR\\
\midrule
MOT$^\dagger$ & 0.01139 & 0.00056 & \multicolumn{1}{c}{n.a.} & 0.001\\
Rule-gating & 0.01168 & 0.00049 & 0.01155 & 0.010\\
Context-dependent & 0.01344 & 0.00081 & 0.01330 & 0.010\\
Value-based & 0.01921 & 0.00081 & 0.01904 & 0.010\\
Rule-gating (Spec A) & 0.02014 & 0.00092 & 0.01999 & 0.010\\
Neural PT & 0.02064 & 0.00085 & 0.02050 & 0.010\\
Neural CPT & 0.02145 & 0.00093 & 0.02134 & 0.010\\
Neural EU & 0.02215 & 0.00085 & 0.02205 & 0.010\\
\bottomrule
\end{tabular}\\[3pt]
{\footnotesize\textit{Note:} $^\dagger$MOT learning rate is selected by test MSE following the published protocol of \citet{peterson2021}. Rule-gating (Spec~A) excludes \textsc{A1}/\textsc{A2}. The full 12-rule specification (Rule-gating) is reproduced from Table~\ref{tab:cv_results} for comparison.}
\end{table}

Figure~\ref{fig:context_complexity_noatt} replicates the complexity mechanism plot under the 10-rule specification. The qualitative patterns persist: removing \textsc{A1}/\textsc{A2} reallocates mass to the structured rules, but directional movements are unchanged. The same holds for risk asymmetry (not shown).

\begin{figure}[h!]
\centering
\begin{subfigure}[t]{1\linewidth}
    \centering
    \includegraphics[width=\linewidth]{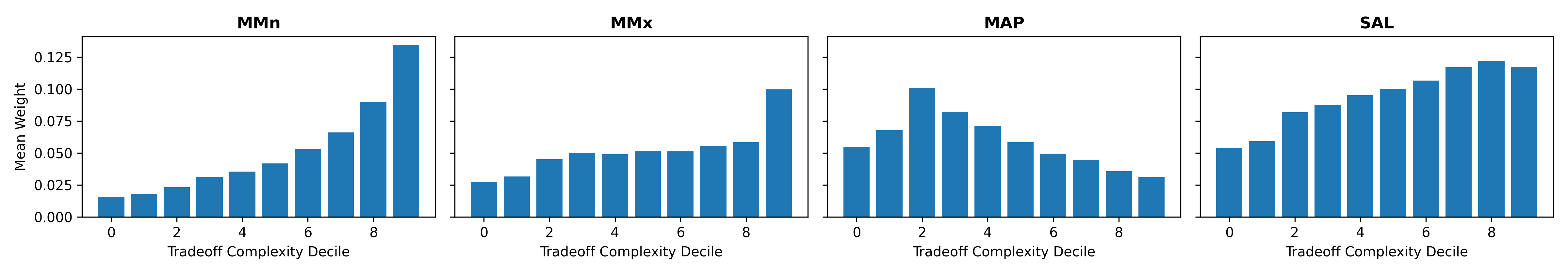}
    \caption{Full library (12 rules).}
\end{subfigure}\hfill
\begin{subfigure}[t]{1\linewidth}
    \centering
    \includegraphics[width=\linewidth]{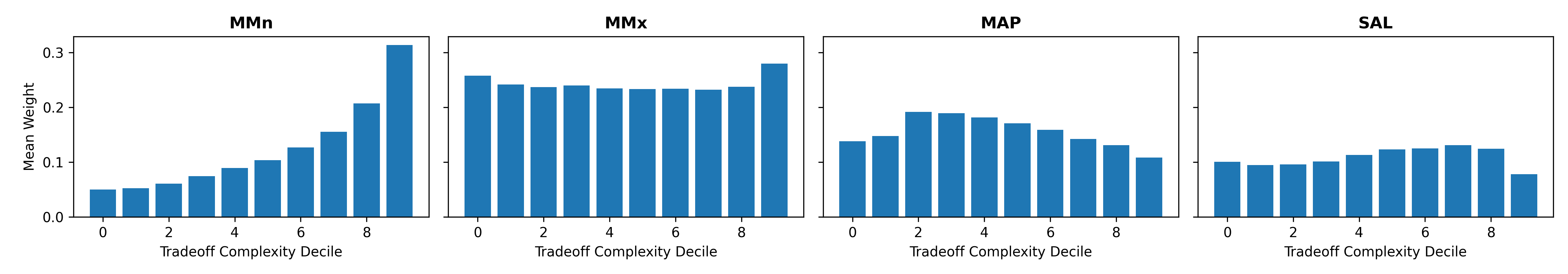}
    \caption{10 rules (no \textsc{A1}/\textsc{A2}).}
\end{subfigure}
\caption{\label{fig:context_complexity_noatt}Mean effective responsibility weights across deciles of tradeoff complexity for \textsc{MMn}, \textsc{MMx}, \textsc{MAP}, \textsc{SAL}: full library vs.\ 10-rule specification.}
\end{figure}

\section{Paired Comparisons}\label{app:paired}

Table~\ref{tab:paired_comparisons} reports split-level comparisons for out-of-sample performance. For each pair of models, we compute the per-split difference in test MSE and report the mean and standard error across the 50 random splits. Because the train/test splits overlap across repetitions, these standard errors should be interpreted as descriptive measures of split-to-split variability rather than formal sampling-based inference.

\begin{table}

\caption{\label{tab:paired_comparisons}Paired comparisons of models against the rule-gating model (fold-level deltas). Positive delta means the comparison model performs worse. SE is split-to-split standard error.}
\centering
\begin{tabular}[t]{lllllr}
\toprule
Model A & vs. & Metric & Mean delta & SE & N folds\\
\midrule
Context-dependent & Rule-gating & test\_mse & -0.00045 & 0.00010 & 50\\
Neural CPT & Rule-gating & test\_mse & 0.00756 & 0.00014 & 50\\
Neural EU & Rule-gating & test\_mse & 0.00826 & 0.00012 & 50\\
Neural PT & Rule-gating & test\_mse & 0.00675 & 0.00013 & 50\\
Rule-gating & Rule-gating & test\_mse & -0.00221 & 0.00006 & 50\\
\addlinespace
Value-based & Rule-gating & test\_mse & 0.00532 & 0.00012 & 50\\
MOT & Rule-gating & test\_mse & -0.00250 & 0.00012 & 50\\
\bottomrule
\end{tabular}
\end{table}

\section{Model-Level Uncertainty Intervals}\label{app:model_ci}

Table~\ref{tab:model_ci} reports split-variability intervals ($\text{mean}\pm 1.96\times\text{SE}$) for each model's out-of-sample metrics, constructed from fold-to-fold variability across the 50 CV splits. These intervals are descriptive summaries of how performance varies across train/test partitions, not formal confidence intervals.

\begin{table}[h!]
\centering
\caption{\label{tab:model_ci}Model-level split-variability intervals from fold-to-fold variability (50 splits).}
\begin{tabular}{lrrrr}
\toprule
Model & Mean MSE & SE & CI low & CI high \\
\midrule
MOT & 0.01139 & 0.00008 & 0.01124 & 0.01155 \\
Rule-gating & 0.01168 & 0.00007 & 0.01154 & 0.01182 \\
Context-dependent & 0.01344 & 0.00012 & 0.01322 & 0.01367 \\
Value-based & 0.01921 & 0.00011 & 0.01898 & 0.01943 \\
Neural PT & 0.02064 & 0.00012 & 0.02040 & 0.02087 \\
Neural CPT & 0.02145 & 0.00013 & 0.02119 & 0.02171 \\
Neural EU & 0.02215 & 0.00012 & 0.02192 & 0.02239 \\
\bottomrule
\end{tabular}\\[3pt]
{\footnotesize\textit{Note:} Intervals computed as mean $\pm$ 1.96 $\times$ SE, where SE = SD / $\sqrt{n_{\text{folds}}}$. These are descriptive split-variability intervals, not formal confidence intervals.}
\end{table}

\section{Ablation Effect Split-Variability Intervals}\label{app:phi_ci}

Table~\ref{tab:phi_ci} reports split-variability intervals for the absolute ablation effect $\Delta\text{MSE}=\text{MSE}^{-f}-\text{MSE}^{\text{full}}$, constructed from fold-to-fold variability. For each rule $f$, we compute the per-fold MSE increase from dropping rule $f$ from the full 12-rule library, then report the mean, standard error, and interval bounds. (The relative index $\phi(f)=\Delta\text{MSE}/\text{MSE}^{\text{full}}$, as defined in Equation~\eqref{eq:phi}, is reported in Table~\ref{tab:decomposition_stability}.)

The leading ablation effects are \textsc{SAL2} and \textsc{SAL}, followed by \textsc{REGmed} and \textsc{REG}. At the other extreme, \textsc{MMa} has a small negative ablation effect (removing it slightly improves MSE), indicating that its contribution is fully substitutable by the remaining rules.

\begin{table}[h!]
\centering
\caption{\label{tab:phi_ci}Ablation effect (absolute MSE change) with split-variability intervals from fold-to-fold variability. Each rule is ablated from the full 12-rule library.}
\begin{tabular}{lrrrr}
\toprule
Rule & Mean $\Delta\text{MSE}$ & SE & CI low & CI high \\
\midrule
\textsc{SAL2} & 0.00061 & 0.00005 & 0.00051 & 0.00071 \\
\textsc{SAL} & 0.00054 & 0.00004 & 0.00046 & 0.00063 \\
\textsc{REGmed} & 0.00021 & 0.00003 & 0.00014 & 0.00028 \\
\textsc{REG} & 0.00014 & 0.00005 & 0.00005 & 0.00023 \\
\textsc{MAP} & 0.00014 & 0.00003 & 0.00007 & 0.00020 \\
\textsc{MMx} & 0.00013 & 0.00004 & 0.00005 & 0.00022 \\
\textsc{MMn} & 0.00011 & 0.00003 & 0.00005 & 0.00017 \\
\textsc{DIS} & 0.00005 & 0.00003 & -0.00001 & 0.00010 \\
\textsc{DISmed} & 0.00000 & 0.00003 & -0.00006 & 0.00006 \\
\textsc{MMa} & -0.00004 & 0.00003 & -0.00009 & 0.00001 \\
\bottomrule
\end{tabular}\\[3pt]
{\footnotesize\textit{Note:} $\Delta\text{MSE}=\text{MSE}^{-f}-\text{MSE}^{\text{full}}$: absolute MSE increase when rule $f$ is dropped (fold-level paired differences). The relative index $\phi(f)=\Delta\text{MSE}/\text{MSE}^{\text{full}}$ is reported in Table~\ref{tab:decomposition_stability}.}
\end{table}

\section{Decomposition Stability}\label{app:decomposition_stability}

Table~\ref{tab:decomposition_stability} reports the stability of the responsibility weights $w_f$ and ablation index $\phi(f)$ across the 50 CV folds. For each rule, we report the mean, standard deviation, and coefficient of variation (CV = SD/mean) across folds. The Spearman rank correlations reported in the table note summarize the stability of the \emph{orderings}: for all $\binom{50}{2}$ fold pairs, we compute the Spearman correlation between the fold-level $w_f$ (or $\phi(f)$) vectors, and report the mean, minimum, and maximum.

\begin{table}[t]
\centering
\caption{\label{tab:decomposition_stability}Decomposition stability across 50 CV folds. For each rule $f$: mean, standard deviation, and coefficient of variation (CV) of the responsibility weight $w_f$ and ablation index $\phi(f)$ across folds.}
\begin{tabular}{lrrrrrr}
\toprule
Rule & Mean $w_f$ & SD $w_f$ & CV $w_f$ & Mean $\phi(f)$ & SD $\phi(f)$ & CV $\phi(f)$ \\
\midrule
\textsc{A1} & 0.2882 & 0.0098 & 0.03 & --- & --- & --- \\
\textsc{A2} & 0.2631 & 0.0092 & 0.03 & --- & --- & --- \\
\textsc{SAL} & 0.0927 & 0.0060 & 0.06 & 0.0471 & 0.0271 & 0.57 \\
\textsc{REG} & 0.0717 & 0.0056 & 0.08 & 0.0123 & 0.0270 & 2.19 \\
\textsc{MAP} & 0.0595 & 0.0047 & 0.08 & 0.0119 & 0.0193 & 1.62 \\
\textsc{MMx} & 0.0500 & 0.0058 & 0.12 & 0.0118 & 0.0263 & 2.23 \\
\textsc{MMn} & 0.0432 & 0.0143 & 0.33 & 0.0095 & 0.0177 & 1.86 \\
\textsc{SAL2} & 0.0385 & 0.0022 & 0.06 & 0.0529 & 0.0300 & 0.57 \\
\textsc{REGmed} & 0.0373 & 0.0085 & 0.23 & 0.0178 & 0.0210 & 1.18 \\
\textsc{DISmed} & 0.0268 & 0.0076 & 0.28 & 0.0005 & 0.0179 & 39.76 \\
\textsc{DIS} & 0.0218 & 0.0104 & 0.48 & 0.0041 & 0.0174 & 4.28 \\
\textsc{MMa} & 0.0073 & 0.0048 & 0.65 & -0.0032 & 0.0157 & 4.98 \\
\bottomrule
\end{tabular}
\\[3pt]{\footnotesize\textit{Note:} Spearman rank correlation of $w_f$ orderings across all fold pairs: mean $\rho=0.930$, range $[0.762, 1.000]$. Spearman rank correlation of $\phi(f)$ orderings: mean $\rho=0.442$, range $[-0.636, 0.952]$.}
\end{table}

\section{Selection Stability}\label{app:selection_stability}

Table~\ref{tab:selection_stability} reports the frequency with which each behavioral rule is selected in the top-$k$ set across the 50 cross-fitted CV folds. Higher frequency indicates more stable selection.

\begin{table}[h!]
\centering
\caption{\label{tab:selection_stability}Cross-fitted selection stability: frequency of each behavioral rule appearing in top-$k$ across 50 folds.}
\begin{tabular}{lrrrr}
\toprule
Rule & Freq.\ top-3 & Freq.\ top-5 & Freq.\ top-7 & Freq.\ top-12 \\
\midrule
\textsc{A1} & 1.00 & 1.00 & 1.00 & 1.00 \\
\textsc{A2} & 1.00 & 1.00 & 1.00 & 1.00 \\
\textsc{DIS} & 0.00 & 0.00 & 0.00 & 1.00 \\
\textsc{DISmed} & 0.00 & 0.00 & 0.00 & 1.00 \\
\textsc{MAP} & 0.00 & 0.96 & 1.00 & 1.00 \\
\textsc{MMa} & 0.00 & 0.00 & 0.00 & 1.00 \\
\textsc{MMn} & 0.00 & 0.00 & 0.00 & 1.00 \\
\textsc{MMx} & 0.00 & 0.04 & 1.00 & 1.00 \\
\textsc{REG} & 0.00 & 1.00 & 1.00 & 1.00 \\
\textsc{REGmed} & 0.00 & 0.00 & 0.44 & 1.00 \\
\textsc{SAL} & 1.00 & 1.00 & 1.00 & 1.00 \\
\textsc{SAL2} & 0.00 & 0.00 & 0.56 & 1.00 \\
\bottomrule
\end{tabular}\\[3pt]
{\footnotesize\textit{Note:} Frequency 1.00 means the rule is selected in every fold.}
\end{table}

\section{Feature-Based Baselines}\label{app:sklearn}

We train two baseline models---a \emph{fractional logit} \citep{papke_wooldridge_1996} and a \emph{gradient-boosted tree ensemble} (GBT, scikit-learn)---on three feature sets: (F1)~summary statistics ($\approx$18 dim), (F2)~rule indicators (24 dim), and (F3)~concatenation ($\approx$42 dim). All use the same 50-fold CV protocol. The fractional logit is a binomial GLM with logit link on the continuous choice frequency $\hat{p}(A)\in[0,1]$; $L_2$ regularization is selected from $\{0, 0.01, 0.1, 1\}$ by inner-fold validation. GBT uses \texttt{max\_depth=4} with the number of estimators selected from $\{100, 200, 500\}$.

\paragraph{Results.} Table~\ref{tab:sklearn} reports results. On rule features alone (F2), rule-gating substantially outperforms both baselines. This comparison isolates the gating mechanism's contribution (not a pure feature comparison, since rule-gating also uses $z(A)$). GBT on combined features (F3) outperforms all structured models but provides no behavioral decomposition.

\begin{table}[t]
\centering
\caption{\label{tab:sklearn}Feature-based baselines: out-of-sample MSE (50 splits). Rules = binary rule indicators; Stats = lottery summary statistics; Combined = concatenation.}
\begin{tabular}{lrrrr}
\toprule
Model & Mean MSE & SD & Mean $\mathrm{MSE}_w$ & Features \\
\midrule
GBT (combined) & 0.00957 & 0.00041 & 0.00947 & F3 \\
GBT (stats) & 0.00997 & 0.00041 & 0.00986 & F1 \\
Fractional logit (combined) & 0.01579 & 0.00073 & 0.01565 & F3 \\
GBT (rules) & 0.02039 & 0.00079 & 0.02013 & F2 \\
Fractional logit (rules) & 0.02346 & 0.00103 & 0.02324 & F2 \\
Fractional logit (stats) & 0.02362 & 0.00099 & 0.02338 & F1 \\
\bottomrule
\end{tabular}\\[3pt]
{\footnotesize\textit{Note:} F1 = summary statistics ($\approx$18 dim); F2 = rule indicators (24 dim); F3 = combined ($\approx$42 dim). All models use the same 50-fold CV protocol as Table~\ref{tab:cv_results}.}
\end{table}

\section{Placebo Rule Library}\label{app:placebo}

We construct a placebo test by randomly reassigning each rule's $(L_f(A),A_f(A))$ pairs across menus within $10$ complexity strata,\footnote{Stratified by support-count quantiles to preserve marginal activity rates per stratum.} preserving marginal statistics but destroying menu-specific content ($50$ repetitions). Table~\ref{tab:placebo} shows that the placebo library produces substantially higher MSE, confirming that the rule indicators carry genuine, menu-specific predictive content.

\begin{table}[t]
\centering
\caption{\label{tab:placebo}Placebo rule library test. The ``Real'' row reports the rule-gating model with the actual rule indicators; the ``Placebo'' row uses stratified-permuted indicators that preserve marginal activity rates per rule per complexity bin but destroy the menu-to-indicator mapping.}
\begin{tabular}{lrrr}
\toprule
Library & Mean MSE & SD & $n$ \\
\midrule
Real (rule-gating) & 0.01168 & 0.00049 & 50 \\
Placebo (permuted) & 0.02197 & 0.00032 & 50 \\
\bottomrule
\end{tabular}\\[3pt]
{\footnotesize\textit{Note:} Placebo indicators are constructed by independently permuting each rule's $(L_f, A_f)$ pairs across menus within quantile bins of menu complexity (number of nonzero outcomes). The gate architecture, training protocol, and all other inputs are identical to the main rule-gating model. SD for the real model is across 50 CV folds; SD for the placebo is across permutations.}
\end{table}

\section{Raw-Encoding Gate Benchmark}\label{app:raw_gate}

We estimate a variant using the raw encoding $\psi(A)\in\mathbb{R}^{40}$ instead of $z(A)$, yielding $40\times|\mathcal{F}|$ gate parameters. The raw-encoding gate achieves MSE $0.0139$ vs.\ $0.0117$ for the baseline---about $19\%$ worse (Figure~\ref{fig:gate_comparison})---indicating that the interpretable summary statistics capture decision-relevant variation without the overfitting induced by higher-dimensional inputs.

\begin{figure}[t]
\centering
\includegraphics[width=0.65\linewidth]{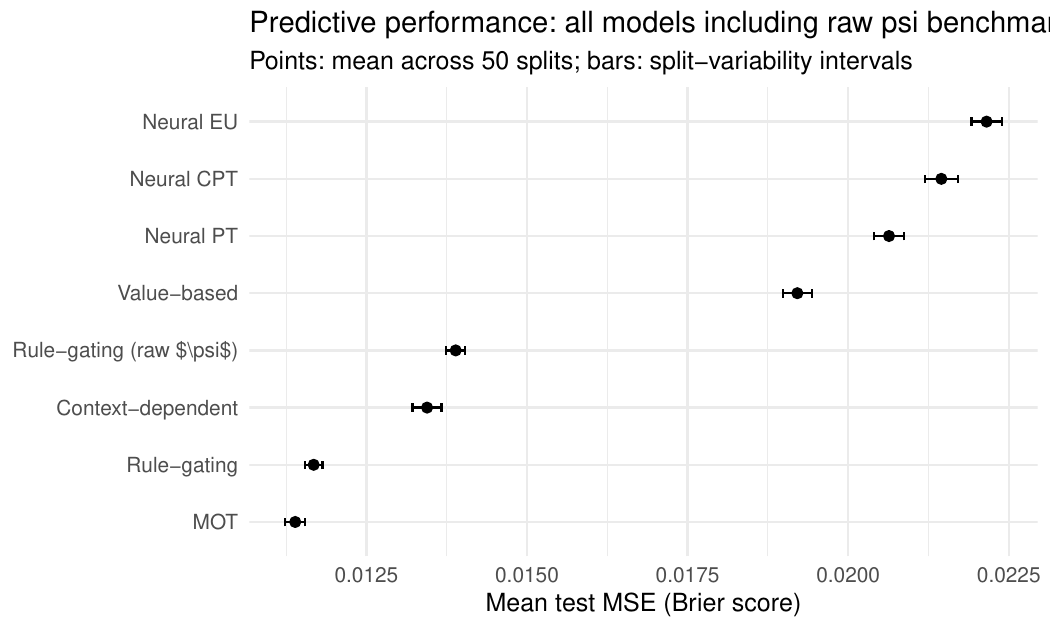}
\caption{\label{fig:gate_comparison}Predictive performance of all models including the raw-encoding gate variant: mean test MSE across 50 random splits with split-variability intervals.}
\end{figure}

\section{Local Identification for High-Dimensional Gate Features}\label{app:jacobian_details}

The global conditions (G1)--(G2) require within-feature replication, which fails when the gate operates on the essentially injective raw encoding $\psi(A)\in\mathbb{R}^{4M}$. For such specifications, identification is assessed locally via Theorem~\ref{thm:local_id_gen}: if the Jacobian $J(\vartheta^\star)$ of the identification map (Section~\ref{app:jacobian}) has full column rank at the fitted parameters, part~(a) guarantees local identification.

In Section~\ref{subsec:rank_empirics}, both diagnostics support identification: the global conditions hold for the $z(A)$ specification, and the Jacobian achieves full column rank at the fitted parameters for both the $z(A)$ and raw-$\psi$ specifications, where ``full column rank'' is evaluated relative to the effective parameter dimension (accounting for the redundancy in $z(A)$ that reduces its dimension from $12$ to $d_{\mathrm{eff}}=11$).

\end{document}